\documentclass[12pt]{article}

\usepackage{amsthm,amsmath,amssymb}
\usepackage[round]{natbib}
\usepackage{mathrsfs}
\usepackage{amsfonts}
\usepackage{enumerate}
\usepackage{latexsym}
\usepackage{multirow}
\usepackage{cite}
\usepackage[pdftex]{color,graphicx}
\usepackage[colorlinks,citecolor=blue,urlcolor=blue]{hyperref}
\newcommand{\argmin}{\operatornamewithlimits{\textrm{argmin}}}

\newcommand{\sbt}{\mathrm{subject\; to }}
\numberwithin{figure}{section}

\newtheorem{prop}{Proposition}[section]
\newtheorem{thm}[prop]{Theorem}
\newtheorem{lemma}[prop]{Lemma}
\newtheorem{cor}[prop]{Corollary}
\newtheorem{remark}[prop]{Remark}

\newcommand{\E}{{\mathbb E}}

\newcommand{\R}{{\mathbb R}}
\renewcommand{\P}{{\mathbb P}}

\newcommand{\B}{\boldsymbol}
\newcommand{\C}{{\mathcal{C}}}
\newcommand{\I}{{\mathcal{I}}}
\newcommand{\D}{{\mathcal{D}}}

\newcommand{\s}{{\mathcal{S}}}

\newcommand{\Z}{{\mathcal{Z}}}
\newcommand{\h}{{\mathcal{H}}}
\newcommand{\N}{{\mathcal{N}}}

\newcommand{\G}{{\mathcal{G}}}

\newcommand{\X}{{\mathcal{X}}}

\newcommand{\eps}{{\epsilon}}

\newcommand{\bt}{\pmb \theta}
\newcommand{\tbd}{\tilde{\pmb \delta}}
\newcommand{\br}{\pmb \rho}
\newcommand{\bd}{\pmb \delta}

\newcounter{rcnt}[section]

\def\argmin{\mathop{\rm argmin}}

\begin{document}

\title{Testing against a linear regression model using ideas from shape-restricted estimation}  
\author{Bodhisattva Sen\footnote{Supported by NSF Grants DMS-1150435 and AST-1107373}  $\hspace{0.03in}$ and Mary Meyer\footnote{Supported by NSF Grant DMS-0905656} \\ Columbia University, New York and Colorado State University, Fort Collins}
\maketitle

\begin{center}{\bf Abstract}\end{center}
A formal likelihood ratio hypothesis test for the validity of a parametric regression function is proposed, using a large-dimensional, nonparametric {\em double cone}  alternative.  For example, the test against a constant function uses the alternative of increasing or decreasing regression functions, and the test against a linear function uses the convex or concave alternative. The proposed test is exact, unbiased and the critical value is easily computed. The power of the test increases to one as the sample size increases, under very mild assumptions -- even when the alternative is mis-specified. That is, the power of the test converges to one for any true regression function that deviates (in a non-degenerate way) from the parametric null hypothesis. We also formulate tests for the linear versus partial linear model, and consider the special case of the additive model. Simulations show that our procedure behaves well consistently when compared with other methods. Although the alternative fit is non-parametric, no tuning parameters are involved. 


\section{Introduction}
Let $\mathbf{Y} := (Y_1,Y_2,\ldots, Y_n) \in \R^n$ and consider the model 
\begin{equation}\label{eq:Mdl}
\mathbf{Y} ={\B\theta_0 + \sigma \B\epsilon}, 
\end{equation}
where $\B\theta_0 \in \R^n$, $\B \eps =(\eps_1,\ldots,\eps_n)$, %$\eps_1,\ldots, \eps_n$ 
are i.i.d.~$G$ with mean 0 and variance 1, and $\sigma >0$.  In this paper we address the problem of testing $H_0: \B \theta_0 \in \s$ where 
\begin{equation}\label{eq:DefC}
\s := \{\B \theta  \in \R^n:\B \theta = \B{X}\B{\beta},\B\beta\in\R^k, \mbox{ for some } k \ge 1 \},
\end{equation}
and $\B{X}$ is a known design matrix.
We develop a test for $H_0$ which is equivalent to the likelihood ratio test with normal errors. To describe the test, let $\I$ be a large-dimensional convex cone (for a quite general alternative) that contains the linear space $\s$, and define the ``opposite'' cone $\D=-\I = \{\B x: -\B x \in \I\}$. We test $H_0$ against $H_1: \B \theta \in \I\cup\D \backslash \s$ and the test statistic is formulated by comparing the projection of $\mathbf{Y}$ onto ${\s}$ with the projection of $\mathbf{Y}$ onto the {\em double cone} $\I\cup\D$.  Projections onto convex cones are discussed in~\citet{silvapulle05}, Chapter~3; see~\citet{rwd} for the specific case of isotonic regression, and~\citet{meyer13} for a cone-projection algorithm. 

We show that the test is unbiased, and that the critical value of the test, for any fixed level $\alpha \in (0,1)$, can be computed exactly (via simulation)  if the error distribution $G$ is known (e.g., $G$ is assumed to be standard normal). If $G$ is assumed to be completely unknown, the the critical value can be approximated via the bootstrap. Also, the test is completely automated and does not involve the choice of tuning parameters (e.g., smoothing bandwidths). More importantly, we show that the power of the test converges to 1, under mild conditions as $n$ grows large, for $\B \theta_0$ not only in the alternative $H_1$ but for ``almost'' all $\B \theta_0 \notin \s$. To better understand the scope of our procedure we first look at a few motivating examples. \newline

{\bf Example 1:} Suppose that $\phi_0: [0,1]\rightarrow \R$ is an unknown function of bounded variation and we are given design points $ x_1< x_2<\dots< x_n$ in $[0,1]$, and data $Y_i$, $i = 1,\ldots, n$, from the model: 
\begin{equation}\label{eq:RegMdl}
  Y_i = \phi_0(x_i) + \sigma \eps_i,
\end{equation}
where $\eps_1, \ldots, \eps_n$ are i.i.d.~mean zero variance 1 errors,  and $\sigma > 0$. Suppose that we want to test that $\phi_0$ is a constant function, i.e., $\phi_0\equiv c$, for some unknown $c \in \R$. We can formulate this as in~\eqref{eq:DefC} with $\B \theta_0 := (\phi_0(x_1), \phi_0(x_2), \ldots, \phi_0(x_n))^\top$ $ \in \R^n$ and $\B X = \mathbf{e} := (1,1,\ldots,1)^\top \in \R^n$. We can take 
$$\I = \{\B \theta \in \R^n: \theta_1 \le \theta_2 \le \ldots \le \theta_n\} $$ to be the set of sequences of non-decreasing real numbers. The cone $\I$ can also be expressed as 
\begin{equation}\label{eq:CvxPoly}
\I =\left\{\B\theta\in\R^n:\B{A}\B\theta\geq \B{0}\right\},
\end{equation}
where the $(n-1)\times n$  constraint matrix $\B{A}$ contains mostly zeros, except $\B{A}_{i,i}=-1$ and $\B{A}_{i,i+1}=1$, and $\s=\{\B\theta\in\R^n:\B{A}\B\theta= \B{0}\}$ is the largest linear space in $\I$. Then, not only can we test for $\phi_0$ to be constant against the alternative that it is monotone, but as will be shown in Corollary~\ref{cor:MonoCons}, the power of the test will converge to $1$, as $n \rightarrow \infty$, for any $\phi_0$ of bounded variation that deviates from a constant function in a non-degenerate way. In fact, we find the rates of convergence for our test statistic under $H_0$ (Theorem \ref{thm:Null}) and the alternative (Theorem \ref{thm:Misspec}). The intuition behind this remarkable power property of the test statistic is that a function is both increasing and decreasing if and only if it is a constant. So, if either of the projections of $\mathbf{Y}$, on $\I$ and $\D$ respectively, is not close to $\s$, then the underlying regression function is unlikely to be a constant.
\begin{figure}
\centerline{\includegraphics[height=2.1in, width=5.7in]{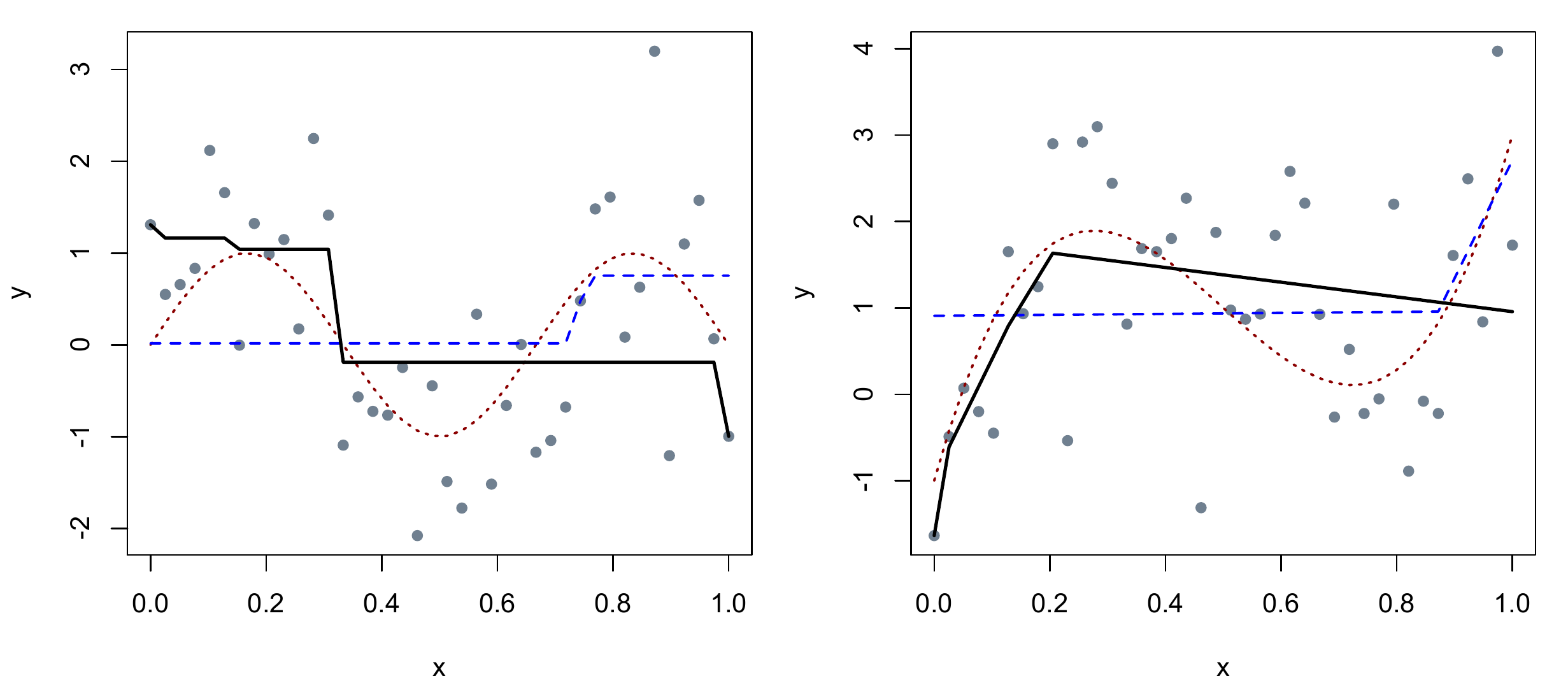} }
\caption{Scatterplots generated from~\eqref{eq:RegMdl} with independent $\N(0,1)$ errors, $\sigma = 1$, and equally spaced $x_i$, where $\phi_0(x)$ is shown as the dotted curve. Left: Increasing (dashed) and decreasing (solid) fits with $\phi_0(x)=\sin(3\pi x)$. Right: Convex (dashed) and concave (solid) fits with $\phi_0(x)=4-6x+40(x-1/2)^3$. }
\label{exampleplot}
\end{figure}

Consider testing against a constant function given the left scatterplot in Fig.~\ref{exampleplot} 
which was generated from a sinusoidal function.  The decreasing fit to the scatterplot represents the projection of $\mathbf{Y}$ onto the double cone, because it has smaller sum of squared residuals than the increasing fit.   Although the true regression function is neither increasing nor decreasing, the projection onto the double cone is sufficiently different from the projection onto $\s$, to the extent that the proposed test rejects $H_0$ at level $\alpha=0.05$.  The power for this test, given the indicated function and error variance, is $0.16$ for $n=50$, rising to $0.53$ for $n=100$ and $0.98$ for $n=200$. Our procedure can be extended to test against a constant function when the covariates are multi-dimensional; see Section~\ref{TestCnst} for the details. \newline

{\bf Example 2:} Consider~\eqref{eq:RegMdl} and suppose that we want to test whether $\phi_0$ is affine, i.e., $\phi_0(x) = a + b x$, $x \in [0,1]$, where $a, b \in \R$ are unknown. This problem can again be formulated as in~\eqref{eq:DefC} where $\mathbf{X} = [\mathbf{e} | \mathbf{e_1}]$ and $\mathbf{e_1} := (x_1, x_2,\ldots, x_n)^\top$. We may use the double cone of convex/concave functions, i.e., $\I$ can be defined as

\begin{equation}\label{eq:CvxCons}
\I = \left \{\B \theta \in \R^n: \frac{\theta_2 - \theta_1}{x_2 - x_1} \le \frac{\theta_3 - \theta_2}{x_3 - x_2} \le \ldots \le \frac{\theta_n - \theta_{n-1}}{x_n - x_{n-1}} \right\},
\end{equation}
if the $x$ values are distinct and ordered.  $\I$ can also be defined by a constraint matrix $\B A$ as in~\eqref{eq:CvxPoly} where  the non-zero elements of the $(n-2)\times n$ matrix $\B{A}$ are $\B{A}_{i,i}=x_{i+2}-x_{i+1}$, $\B{A}_{i,i+1}=x_i-x_{i+2}$, and $\B{A}_{i,i+2}=x_{i+1}-x_i$. We also show that not only can we test for $\phi_0$ to be affine against the alternative that it is convex/concave, but as will be shown in Corollary~\ref{thm:CvxMisspec}, the power of the test will converge to $1$ for any non-affine smooth $\phi_0$.  The second scatterplot in Fig.~\ref{exampleplot} 
is generated from a cubic regression function that is neither convex nor concave, but our test rejects the linearity hypothesis at~$\alpha=0.05$. The power for this test with $n=50$ and the function specified in the plot is $0.77$, rising to 0.99 when $n=100$.

If we want to test against a quadratic function, we can use a matrix appropriate for constraining the third derivative of the regression function to be non-negative. In this case, $\B{A}$ is $(n-3)\times n$ and the non-zero elements are $\B{A}_{i,i}=-(x_{i+3}-x_{i+2})(x_{i+3}-x_{i+1})(x_{i+2}-x_{i+1})$, $\B{A}_{i,i+1}=(x_{i+3}-x_i)(x_{i+3}-x_{i+2})(x_{i+2}-x_i)$, $\B{A}_{i,i+2}=-(x_{i+3}-x_i)(x_{i+3}-x_{i+1})(x_{i+1}-x_i)$, and $\B{A}_{i,i+3}=(x_{i+1}-x_{i})(x_{i+2}-x_{i})(x_{i+2}-x_{i+1})$, for $i=1,\ldots,n-3$. Higher-order polynomials can be used in the null hypothesis by determining the appropriate constraint matrix in a similar fashion.  \newline

{\bf Example 3:} We assume the same setup as in model~\eqref{eq:RegMdl} where now $\{\B{x}_1,\ldots,\B{x}_n\}$ are $n$ distinct points in $\R^d$, for $d \ge 1$. Consider testing for goodness-of-fit of the linear model, i.e., test the hypothesis $H_0: \phi_0$ is affine (i.e., $\phi_0(\B x) = a + {\B b}^\top {\B x}$ for some $a \in \R$ and $\B b \in \R^d$). Define $\B \theta_0 := (\phi_0(\B x_1), \phi_0(\B x_2), \ldots, \phi_0(\B x_n))^\top \in \R^n$, and the model can be seen as a special case of~\eqref{eq:Mdl}. We want to test $H_0: \B \theta_0 \in \s$ where $\s$ is defined as in \eqref{eq:DefC} with $\B X$ being the $n \times (d+1)$ matrix with the $i$-th row $(1, \B x_i)$, for $i=1,\ldots,n$. We can consider $\I$ to be the cone of evaluations of all convex functions, i.e.,
\begin{equation}\label{eq:CvxCone}
\I = \{\B \theta \in \R^n: \psi(\B x_i) = \theta_i, \mbox{ where } \psi:\R^d \rightarrow \R \mbox{ is any convex function}\}.
\end{equation}
The set $\I$ is a closed convex cone in $\R^n$; see~\citet{SS11}. Then, $\D := -\I$ is the set of all vectors that are evaluations of concave functions at~the~data~points. 

We show that under certain assumptions, when $\phi_0$ is any smooth function that is not affine, the power of our test converges to 1; see Section~\ref{MultCvx} for the details and the computation of the test statistic. Again, the intuition behind this power property of the test statistic is that if a function is both convex and concave then it must be affine. 

Over the last two decades several tests for the goodness-of-fit of a parametric model have been proposed; see e.g., \citet{CoxEtAl88}, \citet{Azzalini93}, \citet{Eubank90}, \citet{Hardle93}, \citet{Fan01}, \citet{Stute97}, \citet{Guerre05}, \citet{Christensen10}, \citet{NVK10} and the references therein. Most tests use a nonparametric regression estimator and run into the problem of choosing tuning parameters, e.g., smoothing bandwidth(s). Our procedure does not involve the choice of any tuning parameter. Also, the critical values of most competing procedures need to be approximated using resampling techniques (e.g., the bootstrap) whereas the cut-off in our test can be computed exactly (via simulation) if the error distribution $G$ is assumed known (e.g., Gaussian). \newline

The above three examples demonstrate the usefulness of the test with the double cone alternative. In addition, we formulate a test for the linear versus partial linear model.   For example, we can test the significance of a single predictor while controlling for covariate effects, or we can test the null hypothesis that the response is linear in a single predictor, in the presence of parametrically-modeled covariates. We also provide a test suitable for the special case of additive effects, and a special test against a constant model.

It is worth mentioning that the problem of testing $H_0$ versus a closed convex cone $\I \supset \s$ is well studied; see~\citet{raubertas86} and \citet{rwd}, Chapter 2. Under the normal errors assumption, the null distribution of a likelihood ratio test statistic is that of a mixture of Beta random variables, where the mixing parameters are determined through simulations.
 
The paper is organized as follows: In Section~\ref{Prelim} we introduce some notation and definitions and describe the problem of testing $H_0$ versus a closed convex cone $\I \supset \s$. We describe our test statistic and state the main results about our testing procedure in Section~\ref{Proc}. In Section~\ref{Examples} we get back to the three examples discussed above and characterize the limiting behavior of the test statistic under $H_0$ and otherwise. Extensions of our procedure to weighted regression, partially linear and additive models, as well as testing for a constant function in multi-dimension, are discussed in Section~\ref{Ext}. In Section~\ref{Simul} we illustrate the finite sample behavior of our method and its variants using simulation and real data examples,  and compare it with other competing methods. The Appendix contains proofs of some of the results and other technical details.

\section{Preliminaries}\label{Prelim} 
A set $\C \subset \R^n$ is a {\it cone} if for all $\bt \in \C$ and $\lambda >0$, we have $\lambda \bt \in \C$. If $\C$ is a convex cone then $\alpha_1 \bt_1 + \alpha_2 \bt_2 \in \C$, for any positive scalars $\alpha_1,\alpha_2$ and any $\bt_1, \bt_2 \in \C$. We define the {\it projection} $\B \theta_{I}$ of $\B \theta_0 \in \R^n$ onto ${\cal I}$ as
\begin{equation*}
 \B \theta_{I} := \Pi({\B \theta}_0|\I) \equiv \argmin_{\B \theta \in \I} \|\B \theta_0 - \B \theta\|^2,
\end{equation*}
where $\|\cdot \|$ is the usual Euclidean norm in $\R^n$. The projection is uniquely defined by the two conditions:
\begin{equation}\label{eq:InProj}
	\langle \B \theta_{I}, \B \theta_{0} - \B \theta_{I} \rangle = 0, \quad \quad \langle \B \theta_{0} - \B \theta_{I}, \B \xi \rangle \le 0, \; \mbox{ for all } \B \xi \in \I,
\end{equation}
where for $\mathbf{a} = (a_1,\ldots, a_n)$, $\mathbf{b} = (b_1,\ldots, b_n)$, $\langle \mathbf{a}, \mathbf{b} \rangle = \sum_{i=1}^n a_i b_i$; see Theorem 2.2.1 of \citet{Bertsekas03}. From these it is clear that
\begin{equation}\label{eq:ProjC}
	 \langle \B \theta_{0} - \B \theta_{I}, \B \xi \rangle= 0,  \;\;\mbox{for all}\;\;  \B \xi \in \s.
\end{equation}
Similarly, the projection of $\B \theta_0$ onto $\D$ is
  $\B \theta_{D} := \Pi(\B \theta_0|\D) \equiv \argmin_{\B \theta \in \D} \|\B \theta_0 - \B \theta\|^2$.

Let 
\begin{equation*}
\Omega_I := \I \cap \s^{\perp} \qquad \mbox{and}\qquad \Omega_D := \D \cap \s^{\perp}, 
\end{equation*} 
where $\s^{\perp}$ refers to the orthogonal complement of $\s$. Then $\Omega_I$ and $\Omega_D$ are closed convex cones, and the projection of $\B{y}\in\R^n$ onto $\I$ is the sum of the projections onto $\Omega_I$ and $\s$. The cone $\Omega_I$ can alternatively be specified by a set of {\em generators}; that is, a set of vectors $\B{\delta}_1,\ldots,\B{\delta}_M$ in the cone so that
$$ \Omega_I=\left\{\B\theta\in\R^n:\B\theta=\sum_{j=1}^M b_j \B{\delta}_j, \;\;\mbox{for}\;\;b_j\geq 0,\;\;j=1,\ldots,M\right\}.$$
If the $m \times n$ matrix $\B{A}$ in~\eqref{eq:CvxPoly} has full row rank, then $M=m$ and the generators  of $\Omega_I$ are the columns of $\B{A}^\top(\B{A}\B{A}^\top)^{-1}$. Otherwise, Proposition~1 of \citet{meyer99} can be used to find the generators. The generators of $\Omega_D=\D\cap\s^{\perp}$ are $-\B{\delta}_1,\ldots,-\B{\delta}_M$.

Define the cone {\em polar to} $\I$ as
$\I^o = \left\{ \B{\rho}\in\R^n:\langle\B{\rho},\bt\rangle\leq0,\;\;\mbox{for all}\;\;\bt\in\I\right\}.$
Then, $\I^o$ is a convex cone orthogonal to $\s$.  We can similarly define $\D^o$, also orthogonal to $\s$.
For a proof of the following see \citet{meyer99}.
\begin{lemma}
The projection of $\B{y}\in\R^n$ onto $\I^o$ is the residual of the projection of $\B{y}$ onto $\I$ and vice-versa; also, the projection of $\B{y}$ onto $\D^o$ is the residual of the projection of $\B{y}$ onto $\D$ and vice-versa.
\label{projresid}
\end{lemma}
We will assume the following two conditions for the rest of the paper:

\begin{itemize}
	\item[(A1)] $\s$ is the largest linear subspace contained in $\I$;	
	\item[(A2)] $\Omega_I \subset \D^o$, or equivalently, $\Omega_D \subset \I^o$.

\end{itemize}
Note that if (A1) holds, $\Omega_I$ does not contain a linear space (of dimension one or larger), and the intersection of $\Omega_I$ and $\Omega_D$ is the origin.  Assumption (A2) is needed for unbiasedness of our testing procedure.  For Example~1 of the Introduction, (A2) says that the projection of an increasing vector $\bt$ onto the decreasing cone $\D$ is a constant vector, and vice-versa.   To see that this holds, consider $\bt\in\Omega_I$, i.e., $\sum_{i=1}^n\theta_i=0$ and $\theta_1\leq\cdots\leq\theta_n$, and $\pmb{\rho}\in\D$, i.e., $\rho_1\geq\cdots \geq \rho_n$.  Then defining the partial sums $\Theta_k = \sum_{i=1}^k\theta_i$, we have \[\sum_{i=1}^n\theta_i\rho_i = \theta_1\rho_1 + \sum_{i=2}^n\rho_i(\Theta_i-\Theta_{i-1})\\
= \sum_{i=1}^{n-1} (\rho_i-\rho_{i+1})\Theta_i 
\]
which is negative because $\Theta_i\leq 0$ for $i=1,\ldots,n-1$.  Then by~(\ref{eq:InProj}) the projection of $\bt$ onto $\D$ is the origin.

\subsection{Testing}
We start with a brief review for testing $H_0: \B \theta_0 \in \s$ versus $H_1: \B \theta_0 \in \I \backslash \s$, under the normal errors assumption. The log-likelihood function (up to a constant) is 
$$\ell(\B \theta,\sigma^2) = -\frac{1}{2 \sigma^2} \|\mathbf{Y} - \B \theta\|^2 - \frac{1}{2} n \log \sigma^2.$$ 
After a bit of simplification, we get the likelihood ratio statistic to be $$ \Lambda_I = 2 \left\{ \max_{\bt \in \I, \sigma^2 >0} \ell(\bt,\sigma^2) - \max_{\bt \in \s, \sigma^2 >0} \ell(\bt,\sigma^2) \right \} = n \log \left(\frac{\|\mathbf{Y} - \hat{\B \theta}_S\|^2}{\|\mathbf{Y} - \hat{\B \theta}_I\|^2} \right), $$ where $\hat{\B \theta}_I = \Pi(\mathbf{Y}|\I)$ and $\hat{\B \theta}_S = \Pi(\mathbf{Y}|\s)$. An equivalent test is to reject $H_0$ if 
\begin{eqnarray}\label{eq:TestIso}
T_I (\mathbf{Y})  := \frac{\|\hat{\B \theta}_S - \hat{\B \theta}_I\|^2}{\|\mathbf{Y} - \hat{\B \theta}_S\|^2} 
= \frac{SSE_0-SSE_I}{SSE_0}
\end{eqnarray} 
is large,
where $SSE_0 := \|\mathbf{Y} - \hat{\B \theta}_S\|^2$ is the squared length of the residual of $\mathbf{Y}$ onto $\s$, and $SSE_I := \|\mathbf{Y} - \hat{\B \theta}_I\|^2$ is the squared length of the residual of $\mathbf{Y}$ onto $\I$.   Further,
$SSE_0-SSE_I =\| \Pi(\mathbf{Y}|\Omega_I)\|^2$,
by orthogonality of $\Omega_I$ and $\s$.

Since the null hypothesis is composite, the dependence of the test statistic on parameters under the hypotheses must be assessed. The following result shows that the distribution of $T_I$ is invariant to translations in $\s$ as well as invariant to scale.

\begin{lemma}\label{lem:DistFreeH_0}  For any $\B{s}\in\s$, $\B\epsilon\in\R^n$,
\begin{equation}\label{eq:Invar}
	\Pi(\B\epsilon + \B s|\I) = \Pi(\B\epsilon|\I) + \B s, \quad \mbox{ and } \quad \Pi(\B\epsilon + \B s|\s) = \Pi(\B\epsilon|\s) + \B s.
\end{equation}
Next, consider model~\eqref{eq:Mdl} and suppose that $\B \theta_0 \in \s$. Then the distribution of $T_I(\mathbf{Y})$ is the same as that of $T_I(\B \eps)$.
\end{lemma}
\begin{proof}
To establish the first of the assertions in~\eqref{eq:Invar} it suffices to show that $\Pi(\B\epsilon|\I) + \B s$ satisfies the necessary and sufficient conditions~(\ref{eq:InProj}) for $\Pi(\B\epsilon + \B s|\I)$. Clearly, $\Pi(\B\epsilon|\I) + \B s \in \I$ and $$ \langle \B\epsilon + \B s - (\Pi(\B\epsilon|\I) + \B s), \B \xi \rangle =  \langle \B\epsilon - \Pi(\B\epsilon|\I), \B \xi \rangle \le 0,\qquad \mbox{ for all $\B \xi \in \I$}.$$ Also, 
$$ \langle \Pi(\B\epsilon|\I) + \B s, \B\epsilon+ \B s - (\Pi(\B\epsilon|\I) + \B s)   \rangle =  \langle \Pi(\B\epsilon|\I) + \B s,\B\epsilon  - \Pi(\B\epsilon|\I)  \rangle = 0,$$ as $\langle \Pi(\B\epsilon|\I), \B\epsilon  - \Pi(\B\epsilon|\I)  \rangle = 0$, and $\langle \B s, \B\epsilon  - \Pi(\B\epsilon|\I)  \rangle = 0$ (as $\pm \B s \in \s$). The second assertion in \eqref{eq:Invar} can be established similarly, and in fact, more easily. Also, it is easy to see that, for any $\B\epsilon \in \R^n$,
\begin{equation*}
	\Pi(\sigma\B\epsilon|\I) = \sigma \Pi(\B\epsilon|\I), \quad \mbox{ and } \quad \Pi(\sigma\B\epsilon|\s) = \sigma \Pi(\B\epsilon|\s).
\end{equation*}
Now, using~\eqref{eq:TestIso}, 
$$ T_I(\mathbf{Y}) = \frac{\|\Pi(\sigma \B \eps|\s) - \Pi(\sigma \B \eps|\I)\|^2}{\|\sigma \B \eps - \Pi(\sigma \B \eps|\s)\|^2} = \frac{\|\Pi( \B \eps|\s) - \Pi(\B \eps|\I)\|^2}{\|\B \eps - \Pi(\B \eps|\s)\|^2} = T_I(\B \eps).
$$ This completes the proof.
\end{proof}

\section{Our procedure}\label{Proc}
To test the hypothesis $H_0: \B \theta_0 \in \s$ versus $H_1: \B \theta_0 \in \I \cup\D \backslash \s$, 
we project $\mathbf{Y}$ separately on $\I$  and $\D$ to obtain $\hat{\B \theta}_{I} := \Pi(\mathbf{Y}|\I)$ and $\hat{\B \theta}_{D} := \Pi(\mathbf{Y}|\D)$, respectively. Let $T_I(\mathbf{Y})$ be defined as in \eqref{eq:TestIso} and $T_D(\mathbf{Y})$ defined similarly, with $\hat{\B \theta}_{D}$ instead of $\hat{\B \theta}_{I}$. We define our test statistic (which is equivalent to the likelihood ratio statistic under normal errors) as
\begin{eqnarray}\label{eq:TestS}
	T(\mathbf{Y}) & := & \max \left\{T_I(\mathbf{Y}),T_D(\mathbf{Y})\right\} \nonumber \\ 
	& = & \frac{\max  \left\{ \|\Pi(\mathbf{Y}|\s) - \Pi(\mathbf{Y}|\I)\|^2, \|\Pi(\mathbf{Y}|\s) - \Pi(\mathbf{Y}|\D)\|^2 \right\}}{\|\mathbf{Y} - \Pi(\mathbf{Y}|\s)\|^2}.
\end{eqnarray}
We reject $H_0$ when $T$ is large. 

\begin{lemma}\label{lem:Level}
Consider model~\eqref{eq:Mdl}. Suppose that $\B \theta_0 \in \s$. Then the distribution of $T(\mathbf{Y})$ is the same as that of $T(\B \eps)$, i.e., 
\begin{equation}\label{eq:DistTestS}
T(\mathbf{Y}) = \frac{\max  \left\{ \|\Pi(\B \eps|\s) - \Pi(\B \eps|\I)\|^2, \|\Pi(\B \eps|\s) - \Pi(\B \eps|\D)\|^2 \right\}}{\|\B \eps - \Pi(\B \eps|\s)\|^2} = T(\B \eps).
\end{equation}
\end{lemma}
\begin{proof}
The distribution of $T$ is also invariant to translations in $\s$ and scaling, so if $\B \theta_0 \in \s$, using the same technique as in Lemma~\ref{lem:DistFreeH_0}, we have the desired result. 
\end{proof}
Suppose that $T(\B \eps)$ has distribution function $H_n$. Then, we reject $H_0$ if
\begin{equation}\label{eq:Test}
T(\mathbf{Y}) > c_\alpha := H_n^{-1}(1- \alpha),
\end{equation} 
where $\alpha \in [0,1]$ is the desired level of the test. The distribution $H_n$ can be approximated, up to any desired precision, by Monte Carlo simulations using~\eqref{eq:DistTestS} if $G$ is assumed known; hence, the test procedure described in~\eqref{eq:Test} has exact level $\alpha$, under $H_0$, for any $\alpha \in [0,1]$.

If $G$ is completely unknown, we can approximate $G$ by $G_n$, where $G_n$ is the empirical distribution of the standardized residuals, obtained under $H_0$. Here, by the residual vector we mean $\tilde {\B r} := \mathbf{Y} - \Pi(\mathbf{Y}|\s)$ and by the standardized residual we mean $\hat {\B r} := [\tilde {\B r} - (\B e^\top  \tilde {\B r}/n) \B e]/\sqrt{\mbox{Var}(\hat {\B r})}$, where $\B e:=(1,\ldots, 1)^\top \in \R^n$. Thus, $H_n$ can be approximated by using Monte Carlo simulations where, instead of $G$, we draw i.i.d.~(conditional on the given data) samples from $G_n$.

In fact, the following theorem, proved in Section~\ref{Proofs}, shows that if $G$ is assumed completely unknown, we can bootstrap from any consistent estimator $\hat G_n$ of $G$ and still consistently estimate the critical value of  our test statistic. Note that the conditions required for Theorem~\ref{thm:CutOff} to hold are indeed very minimal and will be satisfied for any reasonable bootstrap scheme, and in particular, when bootstrapping from the empirical distribution of the standardized residuals. 

\begin{thm}\label{thm:CutOff}
Suppose that $\hat G_n$ is a sequence of distribution functions such that $\hat G_n \rightarrow G$ a.s.~and $\int x^2 d\hat G_n(x) \rightarrow  \int x^2 dG(x)$ a.s. Also, suppose that the sequence $\E (n \|\B \eps - \Pi(\B \eps|\s)\|^{-2})$ is bounded. Let $\hat {\B \eps}=(\hat \eps_1,\ldots,\hat \eps_n)$ where $\hat \eps_1,\ldots, \hat \eps_n$ are (conditionally) i.i.d.~$\hat G_n$ and let $D_n$ denote the distribution function of $T(\hat {\B \eps})$, conditional on the data. Define the Levy distance between the distributions $H_n$ and $D_n$ as $$d_L(H_n,D_n) := \inf\{\eta >0: D_n(x - \eta) - \eta \le H_n(x) \le D_n(x + \eta) + \eta, \mbox{ for all } x \in \R\}.$$ Then,
\begin{equation}\label{eq:LevyDist}
	d_L(H_n, D_n) \to 0\;\; \mbox{ a.s.} 
\end{equation}
\end{thm}

It can also be shown that for $\B{\theta}_0\in \I\cup\D$ the test is {\it unbiased}, that is, the power is at least as large as the test size. The proof of the following result can be found in Section~\ref{Unbiasedness}.
\begin{thm}\label{unbiasedness}
Let $\mathbf{Y}_0 := \B{s} + \sigma \B \eps$, for $\B{s}\in\s$, and the components of $\B \eps$ are i.i.d.~$G$. Suppose further that $G$ is a symmetric (around 0) distribution. Choose any $\bt\in\Omega_I$ and let $\mathbf{Y}_1 :=\mathbf{Y}_0+\bt$. Then for any $a>0$,  
\begin{equation}\label{eq:Unbiasedness}
	\P\left(T(\mathbf{Y}_1)>a\right)\geq \P\left(T(\mathbf{Y}_0)>a\right).
\end{equation}
\end{thm}
It is completely analogous to show that the theorem holds for $\bt \in\Omega_D$.  The unbiasedness of the test now follows from the fact that if $\bt_0 \in \I \cup \D \backslash \s$, then $\bt_0 =\B{s}+\bt$ for $\B{s}\in\s$ and either $\bt \in \Omega_I$ or $\bt \in \Omega_D$. In both cases, for $\mathbf{Y} = \bt_0 + \sigma \B \eps$ and $a := H_n^{-1}(1 -\alpha) > 0$, for some $\alpha \in (0,1)$, by~\eqref{eq:Unbiasedness} we have $\P(T(\mathbf{Y}) > a) \ge 1 - (1-\alpha) = \alpha.$

\subsection{Asymptotic power of the test}\label{AsymPower}
We no longer assume that $\bt_0$ is in the double cone unless explicitly mentioned otherwise. We show that, under mild assumptions, the power of the test goes to one, as the sample size increases, if $H_0$ is not true. For convenience of notation, we suppress the dependence on $n$ and continue using the notation introduced in the previous sections. For example we still use $\bt_0, \hat \bt_I, \hat \bt_D, \hat \bt_S$, etc.~although as $n$ changes these vectors obviously change. A intuitive way of visualizing $\bt_0$, as $n$ changes, is to consider $\bt_0$ as the evaluation of a fixed function $\phi_0$ at $n$ points as in \eqref{eq:RegMdl}.

We assume that for any $\bt_0 \in \I$ the projection $\hat \bt_I = \Pi(\mathbf{Y}|\I)$ is consistent in estimating $\bt_0$, under the squared error loss, i.e., for $\mathbf{Y}$ as in model~\eqref{eq:Mdl} and $\bt_0\in\I$, 
\begin{equation}\label{eq:ConsEst}
\|\hat \bt_I  -\bt_0\|^2 = o_p(n).
\end{equation}
The following lemma shows that even if $\bt_0$ does not lie in $\I$,~\eqref{eq:ConsEst} implies that the projection of the data onto $\I$ is close to the projection of $\bt_0$ onto $\I$; see Section~\ref{Proofs} for the proof.
\begin{lemma}\label{lem:conv} 
Consider model~\eqref{eq:Mdl} where now $\bt_0$ is any point in $\R^n$. Let $\bt_I$ be the projection of $\bt_0$ onto $\I$.  If~\eqref{eq:ConsEst} holds,  then $$  \| \hat{\bt}_I -\bt_I\|^2 = o_p(n). $$
Similarly, let $\bt_D$ and $\hat\bt_D$ be the projections of $\bt_0$ and $\mathbf{Y}$ onto $\D$, respectively. Then, if~\eqref{eq:ConsEst} holds, $  \| \hat{\bt}_D -\bt_D\|^2 = o_p(n).$
\end{lemma}
\begin{thm}\label{thm:h0part}
Consider testing $H_0: \bt_0 \in \s$ using the test statistic $T$ defined in~\eqref{eq:TestS} where $\mathbf{Y}$ follows model~\eqref{eq:Mdl}. Then, under $H_0$, $T=o_p(1)$, if~\eqref{eq:ConsEst} holds.
\end{thm}
\begin{proof}
Under $H_0$, $\bt_I = \bt_D = \bt_0$. The denominator of $T$ is $SSE_0$. As $\s$ is a finite dimensional vector space of fixed dimension, $SSE_0/n \rightarrow_p \sigma^2$. The numerator of $T_I$ can be handled as follows. Observe that,
\begin{eqnarray}\label{eq:Simp2} 
\|\hat{\bt}_I-\hat{\bt}_S\|^2 &=& \|\hat{\bt}_I-\bt_I\|^2 + \| \bt_I-\hat{\bt}_S\|^2 +2 \langle\hat{\bt}_I-\bt_I, \bt_I-\hat{\bt}_S\rangle\\
&\le& \|\hat{\bt}_I-\bt_I\|^2 + \| \bt_0 - \hat{\bt}_S\|^2 +2 \| \hat{\bt}_I-\bt_I \| \| \bt_0-\hat{\bt}_S\| \nonumber \\
& = & o_p(n) + O_p(1) + o_p(n), \nonumber
\end{eqnarray}
where we have used Lemma~\ref{lem:conv} and the fact $\| \bt_0 - \hat{\bt}_S\|^2 = O_p(1)$. Therefore, $\|\hat{\bt}_I-\hat{\bt}_S\|^2/n = o_p(1)$ and thus, $T_I = o_p(1)$. An exact same analysis can be done for $T_D$ to obtain the desired result.
\end{proof}

Next we consider the case where the null hypothesis does not hold. If $\bt_S$ is the projection of $\bt_0$ on the linear subspace $\s$, we will assume that the sequence $\{\bt_0\}$, as $n$ grows, is such that
\begin{eqnarray}\label{eq:PowerC}
	\lim_{n \rightarrow \infty} \frac{\max \{\|\bt_I - \bt_S\|, \|\bt_D - \bt_S\| \}}{n} = c,
\end{eqnarray}
for some constant $c >0$.
Obviously, if $\bt_0 \in \s$, then $\bt_S = \bt_I = \bt_D = \bt_0$ and~\eqref{eq:PowerC} does not hold.  If $\bt_0 \in \I \cup \D\backslash \s$, then~\eqref{eq:PowerC} holds if $\|\bt_0-\bt_S\|/n\rightarrow c$, because either $\bt_0=\bt_I$ which implies $\bt_D=\bt_S$ (by (A2)), or  $\bt_0=\bt_D$ which in turn implies $\bt_I=\bt_S$. Observe that~\eqref{eq:PowerC} is essentially the population version of the numerator of our test statistic (see~\eqref{eq:TestS}), where we replace $\mathbf{Y}$ by $\bt_0$. The following result, proved in Section~\ref{Proofs}, shows that if we have a twice-differentiable function that is not affine, then then~\eqref{eq:PowerC} must hold for some $c >0$.

\begin{thm}\label{PowerCond}
Suppose that $\phi_0:[0,1]^d \to \R$, $d \ge 1$, is a twice-continuously differentiable function.  Suppose that $\{\mathbf{X}_i\}_{i=1}^\infty$ be a sequence of i.i.d.~random variables such that $\mathbf{X}_i \sim \mu$, a continuous distribution on $[0,1]^d$. Let $\bt_0 := (\phi_0(\mathbf{X}_1), \ldots, \phi_0(\mathbf{X}_n))^\top$ $\in \R^n$. Let $\I$ be defined as in~\eqref{eq:CvxCone}, i.e., $\I$ is the convex cone of evaluations of all convex functions at the data points. Let $\D := -\I$ and let $\s$ be the set of evaluations of all affine functions at the data points. If $\phi_0$ is not affine a.e.~$\mu$, then~\eqref{eq:PowerC} must hold for some $c >0$.
\end{thm}

Intuitively, we require that $\phi_0$ is different from the null hypothesis class of functions in a non-degenerate way.   For example, if a function is constant except at a finite number of points, then~\eqref{eq:PowerC} does not hold.  Some further motivation for condition~\eqref{eq:PowerC} is given in Section~\ref{PowerTo1}. 

\begin{thm}\label{thm:h1part}
Consider testing $H_0: \bt_0 \in \s$ using the test statistic $T$ defined in~\eqref{eq:TestS} where $\mathbf{Y}$ follows model~\eqref{eq:Mdl}. If~\eqref{eq:ConsEst} and~\eqref{eq:PowerC} hold then $T \rightarrow_p \kappa$, for some $\kappa >0$.
\end{thm}
\begin{proof}
The denominator of $T$ is $SSE_0$. As $\s$ is a finite dimensional vector space of fixed dimension, $SSE_0/n \rightarrow_p \eta$, for some $\eta >0$. The numerator of $T_I$ can be handled as follows. Observe that
\begin{eqnarray*} 
 \| \bt_I-\hat{\bt}_S\|^2 &=& \|\bt_I-\bt_S\|^2 + \|\bt_S - \hat{\bt}_S\|^2 +2 \langle{\bt}_I-\bt_S, \bt_S - \hat{\bt}_S\rangle \\
 &=& \|\bt_I-\bt_S\|^2 + O_p(1)  + o_p(n).
\end{eqnarray*}
Therefore, using~\eqref{eq:Simp2} and the previous display, we get 
\begin{eqnarray*} 
\|\hat{\bt}_I-\hat{\bt}_S\|^2 = o_p(n) + \| \bt_I-\hat{\bt}_S\|^2 =  \|\bt_I-\bt_S\|^2 + o_p(n).
\end{eqnarray*}
Using a similar analysis for $T_D$ gives the desired result.
\end{proof}

\begin{cor}
Fix $0<\alpha <1$ and suppose that~\eqref{eq:ConsEst} and~\eqref{eq:PowerC} hold. Then, the power of the test in~\eqref{eq:Test} converges to 1, i.e., $$ \P(T(\mathbf{Y}) > c_\alpha) \rightarrow 1,  \qquad \mbox{ as } n \rightarrow \infty,$$ where $\mathbf{Y}$ follows model~\eqref{eq:Mdl} and $\alpha \in (0,1)$.
\end{cor}
\begin{proof}
The result immediately follows from Theorems~\ref{thm:h0part} and~\ref{thm:h1part} since $c_\alpha = o_p(1)$ and $T(\mathbf{Y}) \rightarrow_p \kappa >0$, under~\eqref{eq:ConsEst} and~\eqref{eq:PowerC}.
\end{proof}

\section{Examples}\label{Examples}
In this section we come back to the examples discussed in the Introduction. We assume model~\eqref{eq:RegMdl} and that there is a class of functions ${\cal F}$ that is approximated by points in the cone ${\cal I}$. The double cone is thus growing in dimension with the sample size $n$. Then~\eqref{eq:ConsEst} reduces to assuming that the cone is sufficiently large dimensional so that if $\phi_0 \in {\cal F}$, the projection of $\mathbf{Y}$ onto the cone is a consistent estimator of $\phi_0$, i.e.,~\eqref{eq:ConsEst} holds with 
$\B \theta_0 = (\phi_0(x_1),\ldots, \phi_0(x_n))^\top$. Proofs of the results in this section can be found in Section~\ref{Proofs}.

\subsection{Example 1} 
Consider testing against a constant regression function $\phi_0$ in~\eqref{eq:RegMdl}. The following theorem, proved in Section~\ref{Proofs}, is similar in spirit to Theorem~\ref{thm:h0part} but gives the precise rate at which the test statistic $T$ decreases to $0$, under $H_0$.

\begin{thm}\label{thm:Null}
Consider data $\{(x_i,Y_i)\}_{i=1}^n$ from model~\eqref{eq:RegMdl} and suppose that $\phi_0 \equiv~c_0$, for some unknown $c_0 \in \R$. Then $T = O_p(\log n/n)$ where $T$ is defined in \eqref{eq:TestS}.
\end{thm}

The following result, also proved in Section~\ref{Proofs}, shows that for functions of bounded variation which are non-constant, the power of our proposed test converges to 1, as $n$ grows.
\begin{thm}\label{thm:Misspec}
Consider data $\{(x_i,Y_i)\}_{i=1}^n$ from model~\eqref{eq:RegMdl}  where $\phi_0:[0,1] \rightarrow \R$ is assumed to be of bounded variation. Also assume that $x_i \in [0,1]$, for all $i =1,\ldots,n$. Define the design distribution function as 
\begin{equation}\label{eq:DsgnDist}
F_n(s) = \frac{1}{n} \sum_{i=1}^n I(x_i \le s),
\end{equation} 
for $s \in [0,1]$, where $I$ stands for the indicator function. Suppose that there is a continuous strictly increasing distribution function $F$ such that 
\begin{equation}\label{eq:DistConv}
\sup_{s \in[0,1]} |F_n(s) - F(s)| \rightarrow 0{ \;\; \mbox{as} \;\;n\rightarrow\infty}.
\end{equation}  
Also, suppose that $\phi_0$ is not a constant function a.e.~$F$. Then $T {\rightarrow}_p \ c, \mbox{ as } n \rightarrow \infty,$ for some constant $c >0$. 
\end{thm}
\begin{cor}\label{cor:MonoCons}
Consider the same setup as in Theorem~\ref{thm:Misspec}. Then, for $\alpha \in (0,1)$, the power of the test in~\eqref{eq:Test} converges to 1, as $n \to \infty$.
\end{cor}
\begin{proof}
The result immediately follows from Theorems~\ref{thm:Null} and~\ref{thm:Misspec} since $c_\alpha = o_p(1)$ and $T \rightarrow_p c >0$.
\end{proof}

\subsection{Example 2} 
In this subsection we consider testing $H_0: \phi_0$ is affine, i.e., $\phi_0(x) = a + b x$, $x \in [0,1]$, where $a, b \in \R$ are unknown. Recall the setup of Example 2 in the Introduction.

Observe that $\s \subset \I$ as the linear constraints describing $\I$, as stated in \eqref{eq:CvxCons}, are clearly satisfied by any $\B \theta \in \s$. To see that $\s$ is the largest linear subspace contained in $\I$, i.e., (A1) holds, we note that for $\bt\in\I$, $-\bt\in\I$ only if $\bt\in\s$. Assumption (A2) holds if the projection of a convex function, onto the concave cone, is an affine function, and vice-versa. To see this observe that the generators $\bd_1,\ldots,\bd_{n-2}$ of $\Omega_I$ are pairwise positively correlated, so the projection of any $\bd_j$ onto $\D$ is the origin by~(\ref{eq:InProj}).   Therefore the projection of any positive linear combination of the $\bd_j$, i.e.,~any~vector in $\Omega_I$, onto $\D$, is also the origin, and hence projections of vectors in $\I$ onto $\D$ are~in~$\s$.

Next we state two results on the limiting behavior of the test statistic $T$ under the following condition:
\begin{itemize}
	\item[$(C)$] Let $x_1,\ldots, x_n \in [0,1]$. Assume that there exists $c_1,c_2 >0$ such that $c_1/n \le x_i - x_{i-1} \le c_2/n$, for $i=2,\ldots,n$.
\end{itemize}
\begin{thm}\label{thm:CvxNull}
Consider data $\{(x_i,Y_i)\}_{i=1}^n$ from model~\eqref{eq:RegMdl} and suppose that $\phi_0:[0,1] \rightarrow \R$ is affine, i.e., $\phi_0(x) = a+ b x$, for some unknown $a,b \in \R$.  Also suppose that the errors $\eps_i$, $i=1,\ldots, n$, are sub-gaussian. If condition (C) holds then 
$$T = O_p \left( n^{-1} \left(\log \frac{n}{2 c_1}\right)^{5/4} \right), $$ where $T$ is defined in \eqref{eq:TestS}.
\end{thm}
\begin{remark} The proof of the above result is very similar to that of Theorem~\ref{thm:Null}; we now use the fact $ \|\B \theta_0 - \hat{\B  \theta}_I\|^2 = O_p \left((\log \frac{n}{2 c_1})^{5/4} \right)$ (see Remark 2.2 of \citet{GS13}).  
\end{remark}

The following result shows that for any twice-differentiable function $\phi_0$ on $[0,1]$ which is not affine, the power of our test converges to 1, as $n$ grows.
\begin{thm}\label{thm:CvxMisspec}
Consider data $\{(x_i,Y_i)\}_{i=1}^n$ from model~\eqref{eq:RegMdl}  where $\phi_0:[0,1] \rightarrow \R$ is assumed to be twice-differentiable. Assume that condition (C) holds and suppose that the errors $\eps_i$, $i=1,\ldots, n$, are sub-gaussian. Define the design distribution function $F_n$ as in \eqref{eq:DsgnDist} and suppose that there is a continuous strictly increasing distribution function $F$ on $[0,1]$ such that \eqref{eq:DistConv} holds.  Also, suppose that $\phi_0$ is not an affine function a.e.~$F$. Then $T  {\rightarrow}_p \ c,$ for some constant $c >0$. 
\end{thm}
\begin{remark} The proof of the above result is very similar to that of Theorem~\ref{thm:Misspec}; we now use the fact that a twice-differentiable function on $[0,1]$ can be expressed as the difference of two convex functions.
\end{remark}

\begin{cor}\label{cor:1DimCvxCons}
Consider the same setup as in Theorem~\ref{thm:CvxMisspec}. Then, for $\alpha \in (0,1)$, the power of the test based on $T$ converges to 1, as $n \rightarrow \infty.$
\end{cor}

\subsection{Example 3}\label{MultCvx}
Consider model~\eqref{eq:RegMdl} where now $\X := \{\B{x}_1,\ldots,\B{x}_n\} \subset \R^d$, for $d \ge 2$, is a set of $n$ distinct points and $\phi_0$ is defined on a closed convex set $\mathfrak{X} \subset \R^d$. In this subsection we address the problem of testing $H_0: \phi_0$ is affine. Recall the notation from Example 3 in the Introduction. As the convex cone $\I$ under consideration cannot be easily represented as \eqref{eq:CvxPoly}, we first discuss the computation of $\Pi(\B Y| \I)$ and $\Pi(\B Y| \D)$. We can compute $\Pi(\B Y| \I)$ by solving the following (quadratic) optimization  problem:
\begin{equation}\label{eq:CvxLSE}
\begin{aligned}
\mbox{minimize}_{\B\xi_1, \ldots, \B\xi_n;  \B\theta} & \;\;\;\; \| \mathbf{Y} -  \B{\theta} \|^2\\
\sbt \;\;\; &\;\;\;  \theta_{j} + \langle {\Delta}_{ij}, \B{\xi}_{j} \rangle \leq \theta_{i}; \;\;\; i = 1, \ldots, n; \; j = 1, \ldots, n, 
\end{aligned}
\end{equation}
where $\Delta_{ij} := \B{x}_{i} - \B{x}_j \in \R^{d}$, and $\B\xi_i \in \R^d$ and $\B\theta=(\theta_1, \ldots, \theta_{n})^\top  \in \R^{n}$; see \citet{SS11}, \citet{LG12}, \citet{Kuo08}. Note that the solution to the above problem is unique in $\B\theta$ due to the strong convexity of the objective in $\B\theta$. The computation, characterization and consistency of $\Pi(\B Y|\I)$ has been established in \citet{SS11}; also see \citet{LG12}. We use the {\it cvx} package in MATLAB to compute $\Pi(\B Y|\I)$. The projection on $\D$ can be obtained by solving \eqref{eq:CvxLSE} where we now replace the ``$\le$'' in the constraints by ``$\ge$''.

Although we expect Theorem~\ref{thm:CvxMisspec} to generalize to this case, a complete proof of this fact is difficult and beyond the scope of the paper. The main difficulty is in showing that \eqref{eq:ConsEst} holds for $d \ge 2$.  The convex regression problem described in \eqref{eq:CvxLSE} suffers from possible over-fitting at the boundary of Conv$(\X)$, where Conv$(A)$ denotes the convex hull of the set $A$. The norms of the fitted $\hat{\B \xi}_j$'s near the boundary of Conv$(\X)$ can be very large and there can be a large proportion of data points at the boundary of Conv$(\X)$ for $d \ge 2$. Note that \citet{SS11} shows that the estimated convex function converges to the true $\phi_0$ a.s.~(when $\phi_0$ is convex) only on {\it compacts} in the interior of support of the convex hull of the design points, and does not consider the boundary points.

As a remedy to this possible over-fitting we can consider solving the least squares problem over the class of convex functions that are uniformly Lipschitz. For a convex function $\psi: \mathfrak{X} \rightarrow \R$, let us denote by $\partial \psi(\B x)$ the sub-differential (set of all sub-gradients) set at $\B x \in \mathfrak{X}$, and by $\|\partial \psi(\B x)\|$ the supremum norm of vectors in $\partial \psi(\B x)$. 

For $L >0$, consider the class $\tilde \I_L$ of convex functions with Lipschitz norm bounded by $L$, i.e.,
\begin{equation}\label{eq:C_L}
\tilde \I_L := \{\psi: \mathfrak{X} \rightarrow \R|\ \psi \mbox{ is convex},\ \|\partial \psi\|_{\mathfrak{X}} \le L\}.
\end{equation} 
The resulting optimization problem can now be expressed as (compare with \eqref{eq:CvxLSE}):
\begin{equation*}%\label{eq:CvxLipLSE}
\begin{aligned}
\mbox{minimize}_{\B\xi_1, \ldots, \B\xi_n;  \B\theta} & \;\;\;\; \frac{1}{2}\| \mathbf{Y} -  \B{\theta} \|^2\\
\sbt \;\;\; &\;\;\;  \theta_{j} + \langle \B{\Delta}_{ij}, \B{\xi}_{j} \rangle \leq \theta_{i}; \;\;\; i = 1, \ldots, n; \; j = 1, \ldots, n, \\
& \|\B \xi_j\| \le L, \;\;\; j = 1, \ldots, n.
\end{aligned}
\end{equation*}
Let $\hat \bt_{I,L}$ and $\hat \bt_{D,L}$ denote the projections of $\mathbf{Y}$ onto $\I_L$ and $\D_L$, the set of all evaluations (at the data points) of functions in  $\tilde \I_L$ and $\tilde \D_L := -\tilde \I_L$, respectively. We will use the modified test-statistic 
\begin{equation*}%\label{eq:ModTS}
T_L := \frac{\min\{ \|\hat{\B \theta}_S - \hat{\B \theta}_{D,L}\|^2, \|\hat{\B \theta}_S - \hat{\B \theta}_{I,L}\|^2 \}}{\|\mathbf{Y} - \hat{\B \theta}_S\|^2}.
\end{equation*}
Note that in defining $T_L$ all we have done is to use $\hat \bt_{I,L} $ and $\hat \bt_{D,L} $ instead of $\hat \bt_{I} $ and $\hat \bt_{D} $ as in our original test-statistic $T$. In the following we show that \eqref{eq:ConsEst} holds for $T_L$. The proof of the result can be found in Section~\ref{Proofs}.

\begin{thm}\label{thm:MultCvxCons}
Consider data $\{(\B x_i,Y_i)\}_{i=1}^n$ from the regression model $Y_i = \phi_0(\B x_i) + \eps_i$, for $i=1,2,\ldots,n$, where we now assume that (i) $\mathfrak{X} =[0,1]^d$; (ii) $\phi_0 \in \tilde \I_{L_0}$ for some $L_0 >0$; (iii) $\B x_i \in \mathfrak{X}$'s are fixed constants; and (iv) $\epsilon_i$'s are i.i.d.~sub-gaussian errors. Given data from such a model, and letting $\bt_0 = (\phi_0(\B x_1), \phi_0(\B x_2), \ldots, \phi_0(\B x_n))^\top$, we can show that for any $L>L_0$,
\begin{equation}\label{eq:L_2RateLip}
	\|\hat \bt_{I,L} - \bt_0\|^2 = o_p(n).
\end{equation}
\end{thm}

\begin{remark} 
At a technical level, the above result holds because the class of all convex functions that are uniformly bounded and uniformly Lipschitz is totally bounded (under the $L_\infty$ metric) whereas the class of all convex functions is not totally bounded. 
\end{remark}
%{\color{cyan} As $\I_L$ (and $\D_L$) is not a cone (although it is a closed convex set), Lemma~\ref{lem:DistFreeH_0} does not hold for all $\B{s} \in \s$. However, under $H_0$, the distribution of $T_L$, for $L > L_0$ (where $L_0$ is the norm of the true slope of the linear model, under $H_0$), can still be simulated as described in~\eqref{eq:Test}.} 
The next result shows that the power of the test based on $T_L$ indeed converges to $1$, as $n$ grows. The proof follows using a similar argument as in the proof of Theorem~\ref{thm:Misspec}.

\begin{thm}\label{thm:MultCvxPower}
Consider the setup of Theorem~\ref{thm:MultCvxCons}. Moreover, if the design distribution (of the $\B x_i$'s) converges to a probability measure $\mu$ on $\mathfrak{X}$ such that $\phi_0$ is not an affine function a.e.~$\mu$, then $T_L{\rightarrow}_p \ c,$ for some constant $c >0$. Hence the power of the test based on $T_L$ converges to 1, as $n \rightarrow \infty$, for any significance level $\alpha \in (0,1)$.
\end{thm}

\begin{remark}
We conjecture that for the test based on $T$, as defined in~\eqref{eq:TestS}, Theorems~\ref{thm:MultCvxCons} and~\ref{thm:MultCvxPower} will hold  under appropriate conditions on the design distribution, but a complete proof is beyond the scope of this paper.
\end{remark}
\begin{remark}
The assumption that $\mathfrak{X} =[0,1]^d$ can be extended to any compact subset of $\R^d$.
\end{remark}

\section{Extensions}\label{Ext}
\subsection{Weighted regression}\label{WtReg}

If $\B\epsilon$ is mean zero Gaussian with covariance matrix $\sigma^2\B\Sigma$, for a known positive definite $\B\Sigma$, we can readily transform the problem to the i.i.d.~case. If $\B{U}^\top\B{U}$ is the Cholesky decomposition of $\B\Sigma$, pre-multiply the model equation $\mathbf{Y}=\bt_0+\B{\epsilon}$ through by $\B{U}^\top$ to get $\tilde{\mathbf{Y}}=\tilde{\bt}_0+\tilde{\B{\epsilon}}$, where $\tilde{\epsilon}_1,\ldots,\tilde{\epsilon}_n$ are i.i.d.~mean zero Gaussian errors with variance $\sigma^2$.   Then, minimize $\| \tilde{\mathbf{Y}} - \tilde{\bt}\|^2$  over $\tilde{\bt}\in\tilde{\I}\cup\tilde{\D}$, where $\tilde{\I}$ is defined by $\tilde{\B{A}}=\B{A}(\B{U}^\top)^{-1}$.  A basis for the null space $\tilde{\s}$ is obtained by premultiplying a basis for $\s$ by $\B{U}^\top$, and generators for $\tilde{\Omega}_I=\tilde{\I}\cap\tilde{\s}^\perp $ are obtained by premultiplying the generators of $\Omega_I$ by $\B{U}^\top$.  The test may be performed within the transformed model. If the distribution of the error is non-Gaussian we can still standardize $\mathbf{Y}$ as above, and perform our test after making appropriate modifications while simulating the null distribution.

This is useful for correlated errors with known correlation function, or when the observations are weighted.   Furthermore, we can relax the assumption that the $x$ values are distinct, for if the values of $x$ are not distinct, the $Y_i$ values may be averaged over each distinct $x_i$, and the test can be performed on the averages using the number of terms in the average as weights.

\subsection{Linear versus partially linear models}\label{LinVsPrtLM}
We now consider testing against a parametric regression function, with parametrically modeled covariates.  The model is
\begin{equation}\label{eq:partlin}
Y_i=\phi_0(x_i) + \B{z}_i^\top\B{\alpha}+\epsilon_i,\;\;i=1,\ldots,n,
\end{equation}
where $\B{\alpha}$ is a $k$-dimensional parameter vector, $\B{z}_i$ is the $k$-dimensional covariate, and interest is in testing against a parametric form of $\phi_0$, such as constant or linear, or more generally $\bt_0=\B{X}\B{\beta}$ where $\s=\{\bt\in\R^n:\bt=\B{X}\B{\beta}\}$ is the largest linear space in a convex cone $\I$, for which (A1) and (A2) hold. For example, $\B{X}=\mathbf{e}$ can be used to test for the significance of the predictor, while controlling for the effects of covariates $\B{z}$.    If $\B{X}=[\mathbf{e}|\mathbf{e}_1]$, the null hypothesis is that the expected value of the response is linear in $x$, for any fixed values of the covariates.

Accounting for covariates is important for two reasons.  First, if the covariates explain some of the variation in the response, then the power of the test is higher when the variation is modeled.   Second, if the covariates are related to the predictor, confounding can occur if the covariates are missing from the model.

The assumption that the $x_i$ values are distinct is no longer practical; without covariates we could assume distinct $x_i$ values without loss of generality, because we could average the $Y_i$ values at the distinct $x_i$ and perform a weighted regression.  However, we could have duplicate $x_i$ values that have different covariate values.  Therefore, we need equality constraints as well as inequality constraints, to ensure that $\theta_{0i}=\theta_{0j}$ when $x_i=x_j$.   An appropriate cone can be defined as
$\I =  \{\bt\in\R^n:\B{A}\bt\geq\B{0}\;\;\mbox{and}\;\; \B{B}\bt=\B{0}\}$,
where $\s=\{\bt\in\R^n:\bt=\B{X}\B{\beta}\}$ is the largest linear space in $\I$. 

For identifiability considerations, we assume that the columns of $\B{Z}$ and $\B{X}$ together form a linearly independent set, where $\B{Z}$ is the $n\times k$ design matrix whose rows are $\B{z}_1,\ldots,\B{z}_n$.  %For the constrained model, Theorem~2.4 of \citet{meyer13b} states that the linear and constrained components are jointly estimable if no vector in the column space of $\B{Z}$ is also in the space spanned by the generators $\B{\delta}_1,\ldots,\B{\delta}_M$.   
Let ${\cal L}=\s+{\cal Z}$, where ${\cal Z}$ is the column space of $\B{Z}$.   Define $\tbd_j=\B{\delta}_j-\B{P}_L\B{\delta}_j$, for $j=1,\ldots, M$, where $\B{P}_L$ is the projection matrix for the linear space ${\cal L}$ and $\B{\delta}_1,\ldots,\B{\delta}_M$ are the generators of ${\Omega}_I$. %Under the conditions for estimability, the vectors $\tbd_1,\ldots,\tbd_M$ form an irreducible set, so 
We may now define the cone $\tilde{\Omega}_I$ as generated by $\tbd_1,\ldots,\tbd_M$.  Similarly, the generators of $\tilde{\Omega}_D$ are $-\tbd_1,\ldots,-\tbd_M$.   %For details see \citet{meyer13b}.

Define $\B\xi= \B{\theta}_0 +\B{Z}\B{\alpha}$.  Then $H_0:\B\xi\in{\cal L}$ is the appropriate null hypothesis and the alternative hypothesis is $\B\xi\in\tilde{\I}\cup\tilde{\D}\backslash{\cal L}$, where  $\tilde{\I}={\cal L}+\tilde{\Omega}_I$ and $\tilde{\D}={\cal L}+\tilde{\Omega}_D$.   Then ${\cal L}$ is the largest linear space contained in $\tilde{\I}$ or in $\tilde{\D}$, and it is straight-forward to verify that if $\Omega_D \subseteq\I^o$, we also have $\tilde{\Omega}_D\subseteq\tilde\I^o$.   Therefore the conditions (A1) and (A2) hold for the model with covariates, whenever they hold for the cone without covariates. 

\subsection{Additive models}\label{Additive}
We consider an extension of (\ref{eq:partlin}), where 
\[ Y_i=\phi_{01}(x_{1i})+\cdots+ \phi_{0d}(x_{di}) +\B{z}_i^\top\B{\alpha}+\epsilon_i,\] and the null hypothesis specifies parametric formulations for each $\phi_{0j}$, $j=1,\ldots,d$.  %This provides an alternative test against the linear model, which will have better power in some instances when the additivity assumption is valid.  
Let $\theta_{ji}=\phi_{0j}(x_{ji})$, $j=1,\ldots,d$, and $\B\theta=\B\theta_{1}+\cdots+\B\theta_d+\B{z}_i^\top\B{\alpha}\in\R^n$.  The null hypothesis is $H_0:\B\theta_j\in\s_j$, for $j=1,\ldots,d$, or $H_0:\B\theta\in\s$ where $\s=\s_1+\cdots+\s_d+\Z$, and $\Z$ is the column space of the $n\times k$ matrix whose rows are $\B z_1,\ldots,\B z_n$.   
Define closed convex cones $\I_1,\ldots,\I_d$, where $\s_j$ is the largest linear space in $\I_j$.    Then $\I=\I_1+\cdots+\I_d+\Z$ is a closed convex cone in $\R^n$, containing the linear space $\s$.  The projection $\hat{\B\theta}$ of the data $\mathbf{Y}$ onto the cone $\I$ exists and is unique, and \citet{meyer13b} gave necessary and sufficient conditions for identifiability of the components $\hat{\B\theta}_1,\ldots,\hat{\B\theta}_d$, and $\B\alpha$.    When the identifiability conditions hold, then $\s$ is the largest linear space in $\I$.

Define $\D_j :=-\I_j$ for $i=1,\ldots,d$, and $\D=\D_1+\cdots+\D_d+\Z$.   Then $\I\cup\D$ is a double cone, and we may test the null hypothesis $H_0:\B\theta\in\s$ versus $H_a:\B\theta\in \I\cup\D\backslash\s$ {using the test statistic~\eqref{eq:TestS}}.   However, we may like to include in the alternative hypothesis the possibility that, say, $\B\theta_{1}\in\I_1$ and $\B\theta_{2}\in\D_2$. Thus, for $d=2$, we would like the alternative set to be the quadruple cone defined as the union of four cones: $\I_1+\I_2$, $\I_1+\D_2$, $\D_1+\I_2$, and $\D_1+\D_2$.  Then $\D_1+\D_2$ is the cone opposite to $\I_1+\I_2$, and $\D_1+\I_2$ is the cone opposite to $\I_1+\D_2$, and the largest linear space contained in any of these cones is $\s=\s_1+\s_2+\Z$.   For arbitrary $d\geq 1$, the multiple cone alternative has $2^d$ components; call these $\C_1,\ldots,\C_{2^d}$.  The proposed test involves projecting $\mathbf{Y}$ onto each of the $2^d$ combinations of cones, and 
\[T(\mathbf{Y}) = \frac{\max_{j=1,\ldots,2^d}  \left\{  \|\Pi(\mathbf{Y}|\s) - \Pi(\mathbf{Y}|\C_j)\|^2\right\}}{\|\mathbf{Y} - \Pi(\mathbf{Y}|\s)\|^2}.
\]
using the smallest sum of squared residuals in $T(\mathbf{Y})$.  The distribution of the test statistic is again invariant to scale and translations in $\s$, so for known error distribution $G$, the null distribution may be simulated to the desired precision.

This provides another option for testing against the linear model, that is different from the fully convex/concave alternative of Example~3, but requires the additional assumption of additivity.    It also provides tests for more specific alternatives: for example, suppose that the null hypothesis is $E(y)=\beta_0+\beta_1x_1+\beta_2x_2+\beta_3x_2^2+\beta_4z$, where $z$ is an indicator variable.  If we can assume that the effects are additive, then we can use the cone $\I_1$ of convex functions and the cone $\I_2$ of functions with positive third derivative as outlined in Example~2 of the Introduction. If the additivity assumptions are correct, this quadruple cone alternative might provide better power than the more general, fully convex/concave alternative.

\subsection{Testing against a constant function}\label{TestCnst}
The traditional $F$-test for the parametric least-squares regression model has the null hypothesis that none of the predictors is (linearly) related to the response. For an $n\times p$ full-rank design matrix, the $F$ statistic has null distribution $F(p-1,n-p)$. To test against the constant function when the relationship of the response with the predictors is unspecified, we can turn to our cone alternatives.   

Consider model~(\ref{eq:RegMdl}) where the predictor values are ${\cal X} =\{\B{x}_1, \B{x}_2,\ldots, \B{x}_n\} \subset \R^d$, $d \ge 1$. A cone that contains the one-dimensional null space of all constant vectors is defined for multiple isotonic regression using a partial order on $\mathcal{X}$. That is, $\B{x}_i\preceq \B{x}_j$ if $\B{x}_i\leq \B{x}_j$ holds coordinate-wise.  Two points $\B{x}_i$ and $\B{x}_j$ in  ${\cal X}$  are comparable if either $\B{x}_i\preceq\B{x}_j$ or $\B{x}_j\preceq\B{x}_i$.   Partial orders are reflexive, anti-symmetric, and transitive, but differ from complete orders in that pairs of points are not required to be comparable.  The regression function $\phi_0$ is isotonic with respect to $\preceq$ on ${\cal X}$ if $\phi_0(\B{x}_i)\leq\phi_0(\B{x}_j)$ whenever $\B{x}_i\preceq\B{x}_j$, and $\phi_0$ is anti-tonic if $\phi_0(\B{x}_i)\geq\phi_0(\B{x}_j)$ whenever $\B{x}_i\preceq\B{x}_j$.

In Section~\ref{TestCons} we show that assumptions (A1) and (A2) hold for the double cone of isotonic and anti-tonic functions. However, the double cone for multiple isotonic regression is unsatisfactory because if one of the predictors reverses sign, the value of the statistic~(\ref{eq:TestS}) (for testing against a constant function) also changes.   For two predictors, it is more appropriate to define a quadruple cone, considering pairs of increasing/decreasing relationships in the partial order.
For three predictors we need an octuple cone, which is comprised of four double-cones. See Section~\ref{TestCons} for more details and simulations results.

\section{Simulation studies}\label{Simul}
In this section we investigate the finite-sample performance of the proposed procedure based on $T$, as defined in \eqref{eq:TestS}, for testing the goodness-of-fit of parametric regression models.  We consider the case of a single predictor, the test against a linear regression function with multiple predictors, and the test of linear versus partial linear model, comparing our procedure with competing methods.  In all the simulation settings we assume that the errors are Gaussian. Overall our procedure performs well; although for some scenarios there are other methods that are somewhat better, none of the other methods has the same consistent good performance. Our procedure, being an exact test, always gives the desired level of significance, whereas other methods have inflated test size in some scenarios and are only approximate.   Further, most of the other methods depend on tuning parameters for the alternative fit.

The goodness-of-fit of parametric regression models has received a lot of attention in the statistical literature. \citet{StuteEtAl98} used the empirical process of the regressors marked by the residuals to construct various omnibus goodness-of-fit tests. Wild bootstrap approximations were used to find the cut-off of the test statistics. We denote the two variant test statistics -- the Kolmogorov-Smirnov type and the Cram\'{e}r-von Mises type -- by $S_1$ and $S_2$, respectively. We implement these methods using the ``IntRegGOF'' library in the R package.

\citet{Fan01} proposed a lack-of-fit test based on Fourier transforms; also see \citet{Christensen10} for a very similar method. The main drawback of this approach is that the method needs a reliable estimator of $\sigma^2$ to compute the test-statistic, and it can be very difficult to obtain such an estimator under model mis-specification. We present the power study of the adaptive Neyman test ($T_{AN,1}^*$; see equation (2.1) of \citet{Fan01}) using the  known  $\sigma^2$ (as a gold standard) and an estimated $\sigma^2$. We denote this method by $FH$.  %When using an estimate of $\sigma^2$, as in equation (2.10) of \citet{Fan01}, we got poor results. 

\citet{PS06} proposed an easy-to-implement single global procedure for testing the various assumptions of a linear model. The test can be viewed as a Neyman smooth test and relies only on the standardized residual vector. We implemented the procedure using the ``gvlma'' library in the R package and denote it by $PS$.

\subsection{Examples with a one-dimensional predictor}
Proportions of rejections for 10,000 data sets simulated from $Y_i=\phi_0(x_i)+\epsilon_i$, $i=1,\ldots,n=100$, are shown in Fig.~\ref{quadpowerplot}. In the first plot, the power of the test against a constant function is shown when the true regression function is $\phi_0(x) = 10a(x-2/3)^2_+$, where $(\cdot)_+=\max(\cdot,0)$, and the effect size $a$ ranges from $0$ to $7$.   The alternative for our proposed test (labeled $T$ in the figures) is the increasing/decreasing double cone, and the power is compared with the $F$-test with linear alternative, and the FH test with both known and estimated variance.

In the second plot, power of the test against a linear function is shown for the same ``ramp'' regression function $\phi_0(x) = 10a(x-2/3)^2_+$.  
The alternative for the proposed test is the convex/concave double cone, and the power is compared with the $F$-test with quadratic alternative, the FH test with known and estimated variance, $S_1$ and $S_2$, and the PS test.   As with the test against a constant function, ours has better power and the FH test with estimated variance has inflated test size.

Finally, we consider the null hypothesis that $\phi_0$ is quadratic, and the true regression function is $\phi_0(x)=a\exp(3x-2)$.   The double-cone alternative is as given in Example~2 of the Introduction.  The $S_2$ test has slightly higher power than ours in this situation, and the PS test has power similar to the $S_1$ test.   The FH test with known variance has a small test size compared to the target, and low power.

\begin{figure}
\centerline{\includegraphics[height=1.8in, width=5.6in]{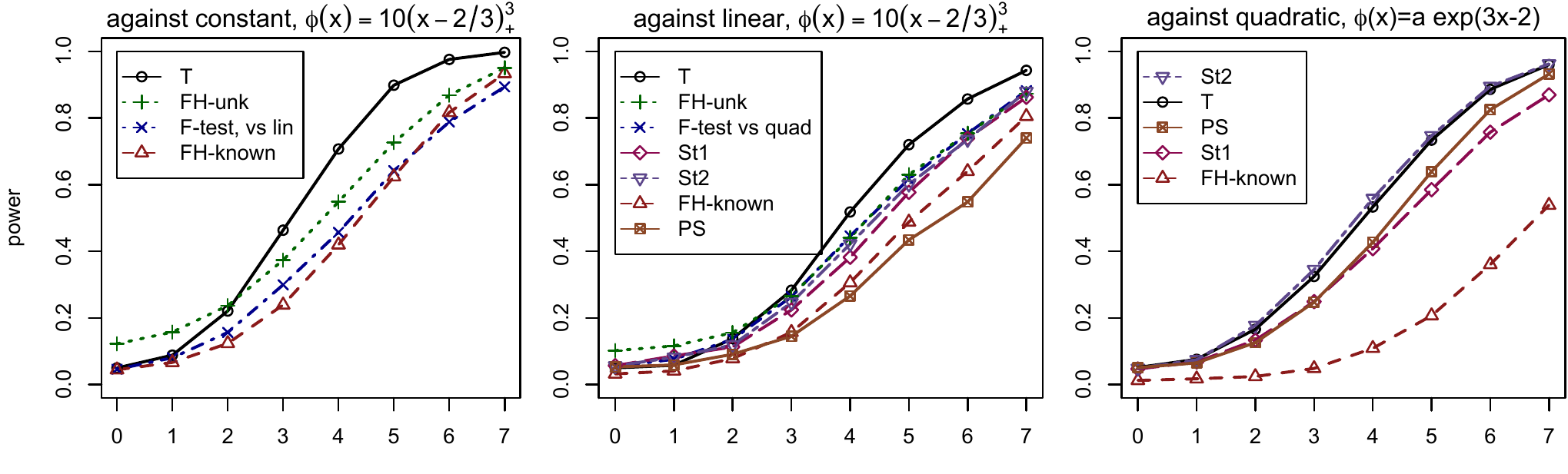} }
\caption{Power function for test against parametric models with one dimensional predictor, $n=100$ observations, equally spaced $x\in[0,1]$ and $\sigma^2=1$.}
\label{quadpowerplot}
\end{figure}

\subsection{Testing against the linear model}
%The goodness-of-fit of a parametric regression model has received a lot of attention in the statistical literature. In this subsection we compare our procedure with six competing methods. 

We consider two data generating models. Model 1 is adapted from \citet{StuteEtAl98} (see Model 3 of their paper) and can be expressed as
%\begin{eqnarray}\label{eq:SimMdl1}
	$Y = 2 + 5 X_{1} -  X_{2} + a X_1 X_2 + \eps$, %\nonumber
%\end{eqnarray}
with covariate $(X_1, \ldots , X_{d})$, where $X_1, \ldots, X_{d}$ are i.i.d.~Uniform$(0,1)$, and $\eps$ is drawn from a normal distribution with mean 0. \citet{StuteEtAl98} used $d = 2$ in their simulations but we use $d = 2, 4$.  Model 2 is adapted from \citet{Fan01} (see Example 4 of their paper) and can be written as
%\begin{eqnarray}\label{eq:SimMdl2}
	$Y = X_1 + a X_2^2 + 2 X_4 + \eps$, % \nonumber
%\end{eqnarray}
where $(X_1, X_2, X_3, X_4)$ is the covariate vector. The covariates $X_1, X_2, X_3$ are normally distributed with mean 0 and variance 1 and pairwise correlation 0.5. The predictor $X_4$ is binary with probability of ``success'' 0.4 and independent of $X_1, X_2$ and $X_3$. Random samples of size $n = 100$, are drawn from Model 1 (and also from Model 2) and a multiple linear regression model is fitted to the samples, without the interaction $X_1 X_2$ term ($X_2^2$ term). Thus, the null hypothesis holds if and only if $a = 0$. In all the following $p$-value calculations, whenever required, we use 1000 bootstrap samples to estimate the critical values of the tests.  {For models~1 and~2 we implement the fully convex/concave double-cone alternative and denote the method by $T$. For model~2 we also implement the octuple cone alternative under the assumption that the effects are additive, and treating $X_4$ as a parametrically modeled covariate (as described in Section~\ref{Additive}). We denote this method by $T_2$.}
%\begin{table}\label{tab:G1}
%\centering 
%\begin{tabular}{|c||c|c|c|c|c|c|c|c|c|c|c|}
%  \hline 
%  $a$ & 0 & 0.5 & 1 & 1.5 & 2 & 3 & 4 & 5 & 7 & 10 \\ \hline   \hline
%     $T$ & 5 & 7 &11 & 11 &14 & 28 & 47 & 65 & 95 &100 \\ \hline
%   $K$ & 6 &  5 &  8 &  14 &  23 &  44 &  67 &  90 &  97 &  100 \\ \hline
%   $S_1$ & 5 & 5 & 6 &  9 & 12 &  20 &  36 &  54 &  82 & 98 \\ \hline
%   $S_2$ & 5 & 6 & 8 & 14 & 20 & 34 &  62 & 85 & 98 & 100  \\ \hline 
%   $F$ & 7 & 8 &  8 & 10 & 9 & 13 & 17 & 24 & 50 & 92  \\ \hline
%   $G$ & 5 & 6 &  6 & 4 & 5 & 7 &  9 &  17 &  56 & 92 \\ \hline
%%   $S_P$ & 5 & 6 & 5 & 4 &  6 & 8 & 9 & 20 & 55 & 92 \\ \hline
%   $L$ & 10 & 14 &  14 & 18 & 23 &  30 &  48 &  66 &  87 & 99 \\ \hline
%\end{tabular}
%	\caption{Percentage of times Model 1 was rejected when $\alpha = 0.05$, $d = 2$ and $n =100$.}
%\end{table}

We also implement the generalized likelihood ratio test of \citet{FJ07}; see equation (4.24) of their paper (also see \citet{FJ05}). The test computes a likelihood ratio statistic, assuming normal errors, obtained from the parametric and nonparametric fits. We denote this method by $L$. As the procedure involves fitting a smooth nonparametric model, it involves the delicate choice of smoothing bandwidth(s). We use the ``np'' library in the R package to compute the nonparametric kernel estimator with the optimal bandwidth being chosen by the ``npregbw'' function in that package. This procedure is similar in spirit to that used in \citet{HM93}. To compute the critical value of the test we use the wild bootstrap method.

\begin{figure}
\centerline{\includegraphics[height=1.8in, width=5.6in]{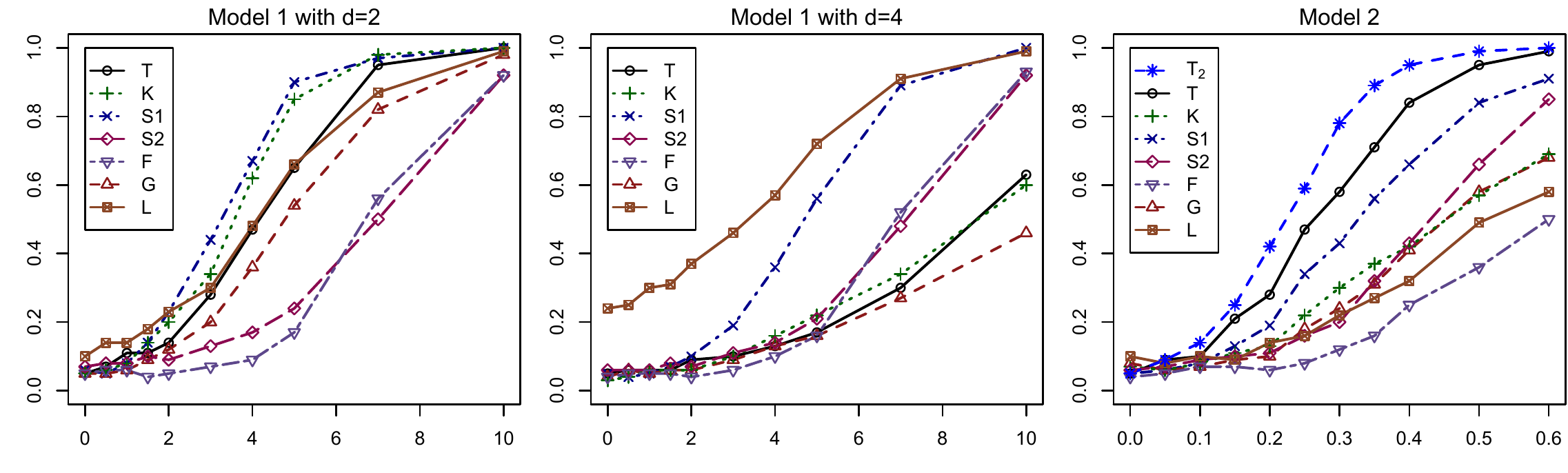}}
\caption{Power function for test against linear models for Models 1 and 2 with $n=100$ observations and $\sigma^2=1$.}
\label{LMpowerplot}
\end{figure}

%\begin{figure}
%\centerline{\includegraphics[height=1.8in, width=5.6in]{TestLM2.pdf}}
%\caption{Power function for test against linear models for Models 1 and 2 with $n=100$ observations and $\sigma^2=1$.}
%\label{LMpowerplot}
%\end{figure}

%\begin{table}\label{tab:G2}
%\centering 
%\begin{tabular}{|c||c|c|c|c|c|c|c|c|c|c|c|}
%  \hline 
%  $a$ & 0 & 0.5 & 1 & 1.5 & 2 & 3 & 4 & 5 & 7 & 10 \\ \hline   \hline
%       $T$ & 5 & 6 & 6 & 6 & 9  & 10  & 13  & 17 & 30 & 63 \\ \hline
%   $K$ & 5 & 4 & 5 & 7 &  10 &  19  & 36 & 56 & 89 &  100 \\ \hline
%   $S_1$  & 4 & 6 & 5 & 6 & 6 & 9 & 13 & 16 &  27 & 46 \\ \hline
%   $S_2$  & 3 & 4 &  6 & 6 & 6 & 10 & 16 & 22 & 34 & 60  \\ \hline 
%   $F$ & 6 & 6 & 6 & 8 & 7 &11 & 14 & 21 & 48 & 92 \\ \hline
%   $G$ & 4 & 5 & 5 & 5 & 4 & 6 &  10 & 16 &  52 &  93 \\ \hline
%%   $S_P$ & 4 &  6 & 7 &  5  & 5 &  6 & 10 & 19 & 53 &  92 \\ \hline
%   $L$ & 24 & 25 & 30 & 31 &  37 &  46 & 57 & 72 & 91 &  99 \\ \hline
%\end{tabular}
%	\caption{Percentage of times Model 1 was rejected when $\alpha = 0.05$, $d = 4$ and $n =100$.}
%\end{table}

We also compare our method with the recently proposed goodness-of-fit test of a linear model by \citet{SS13}. Their procedure assumes the independence of the error and the predictors in the model and tests for the independence of the residual (obtained from the fitted linear model) and the predictors. %The residual, when the model is mis-specified, crucially depends on the predictors and thus yields a high value of the test-statistic, whereas under $H_0$ the residual and the predictors are approximately independent. 
The critical value of the test is computed using a bootstrap approach. We denote this method by $K$. %In fact, we borrow the same simulation setting used in \citet{SS13} and, for convenience, report some of the same numbers (for the competing methods).   

From Fig.~\ref{LMpowerplot} it is clear that our procedure overall has good finite sample performance compared to the competing methods. Note that as $a$ increasing, the power of our test monotonically increases in all problems. As expected, $S_1$ and $S_2$ behave poorly as the dimension of the covariate increases. The method $L$ is anti-conservative and hence shows higher power in some scenarios. It is also computationally intensive, especially for higher dimensional covariates. For model~2, both the fully convex/concave and the octuple-cone additive alternative perform quite well compared to the other methods.

\subsection{Linear versus partial linear model}
We compare the power of the test for linear versus partial linear model~(\ref{eq:partlin}), with $H_0:\phi_0$ is affine, including a categorical covariate with three levels.    The $n=100$ $x$ values are equally spaced in $[0,1]$.   We compare our test with the convex/concave alternative to the standard $F$-test using quadratic alternative, and the test for linear versus partial linear from \citet{Fan01}, Section 2.4 (labeled FH).  Two versions of the FH test are used; the first version uses an estimate of the model variance and the second assumes the variance is known.   

The first two plots in Fig.~\ref{lincovpowerplot} show power for $a=0,1,\ldots,6$ when the true function is $\phi_0(x) = 3ax^2+x$, and the target test size is $\alpha=.05$.   In the first plot, the values of the covariate are generated independently of $x$, and for the second, the predictors are related, so that the categorical covariate is more likely to have level=1 when $x$ is small, and is more likely to have level=3 when $x$ is large.  Here the $F$-test is the gold standard, because the true model satisfies all the assumptions.    The proposed test performs similarly to the FH test with known variance; for the unknown variance case the power for the FH test is larger but the test size is inflated.

In the third and fourth plots, the regression function is $\phi_0(x) =20a(x-1/2)^3+x$.   The $F$-test is not able to reject the null hypothesis because the alternative is incorrect, but the true function is also not contained in the double-cone (convex/concave) of the proposed method.   However, the proposed test still compares well with the FH test, especially when the predictors are correlated.

\begin{figure}
\centerline{\includegraphics[height=1.6in, width=5.6in]{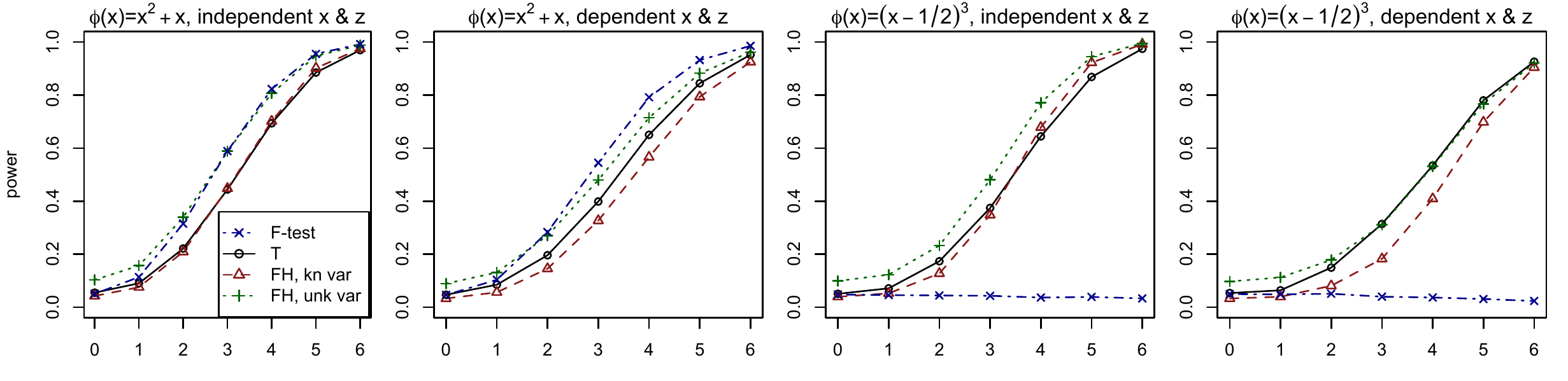} }
\caption{Power for test against affine $\phi_0$ with $n=100$ observations with equally spaced $x\in[0,1]$ and $\sigma^2=1$. The covariate $z$ is categorical with three levels.}
\label{lincovpowerplot}
\end{figure}

\subsection{Testing against constant function, with covariates}

Testing the significance of a predictor while controlling for covariate effects can be accomplished using the partial linear model~(\ref{eq:partlin}), with $H_0:\phi_0(x)\equiv c$, for some unknown $c$, using the double cone alternative for monotone $\phi_0$.  Our method is compared with the standard $F$-test with linear alternative and the FH test, both with known and unknown variance as in the previous subsection.  The first two plots of Fig.~\ref{constcovpowerplot} display power for 10,000 simulated data sets from $\phi_0(x)=ax$, for $n=100$ $x$ values equally spaced in $[0,1]$, with $a$ ranging from 0 to 3. The power of the proposed test is close to the gold-standard $F$-test, and the FH with unknown variance again has inflated test size.   When the predictors are related, the FH test has unacceptably large test size.   

In the third and fourth plots, data were simulated using $\phi_0(x)=a\sin(3\pi x)$, with $a$ ranging from 0 to 1.5.   The true $\phi_0$ is not in the alternative set for either the $F$-test or the proposed test; however the proposed test can reject the null consistently for higher values of $a$. In this scenario, the size of FH test with known variance is inflated for correlated predictors..

\begin{figure}
\centerline{\includegraphics[height=1.5in, width=5.6in]{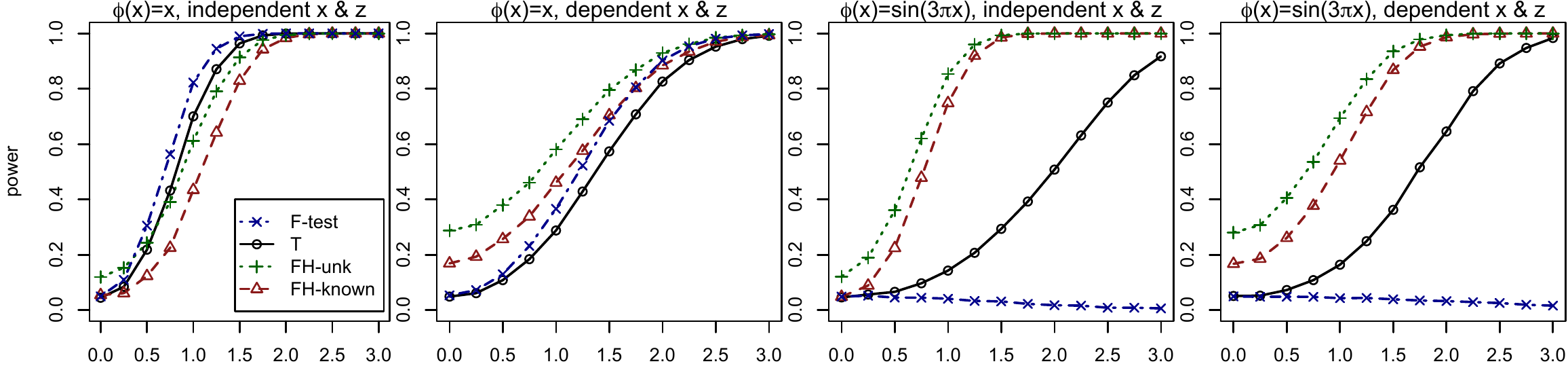} }
\caption{Power for test against constant $\phi_0$ with $n=100$ observations with equally spaced $x\in[0,1]$ and $\sigma^2=1$. The covariate $z$ is categorical with three levels.}
\label{constcovpowerplot}
\end{figure}

\subsection{Real data analysis}

{\bf Data example 1:} We study the well-known Boston housing dataset collected by \citet{HR78} to study the effect of air pollution on real estate price in the greater Boston area in the 1970s. The data consist of 506 observations on 16 variables, with each observation pertaining to one census tract. We use the version of the data that incorporates the minor corrections found by \citet{GP96}. Our procedure, assuming normal errors, yields a $p$-value of essentially 0 and rejects the linear model specification, as used in \citet{HR78}, while the method of \citet{StuteEtAl98} yields a $p$-value of more than 0.2. The method of \citet{SS13} also yields a highly significant $p$-value. \newline

\noindent {\bf Data example 2:} We consider the {\tt Rubber} data set, found in the {\tt R} package {\tt MASS}, representing ``accelerated testing of tyre rubber''.  The response variable is the abrasion loss in gm/hr, with two predictors of loss: the hardness in Shore units, and the tensile strength in kg/sq m.
A linear regression (fitting a plane to the data) provides $R^2=.84$ and the usual residual plots do not provide evidence against linearity.  Further, the Stute tests provide $p$-values of .39 and .18, respectively, and the Pe\~na and Slate test against the linear model provides $p=.92$.   The quadruple cone alternative of Section~\ref{Additive} provides $p=.047$, although the fully convex/concave double cone alternative does not reject at $\alpha=.05$.  Some insight into the true function can be found by fitting the constrained additive model, using the reasonable assumption that the expected response is decreasing in both predictors.  The fit is roughly linear in ``hardness'' but is more like a sigmoidal or step function in ``tensile strength''; these components are shown in Fig.~\ref{fig:rub}.  Although neither fit is convex or concave, the quadruple cone method can still detect the departure from linearity.
 \newline
 
\begin{figure}
\centerline{\includegraphics[height=1.7in,width=4.0in]{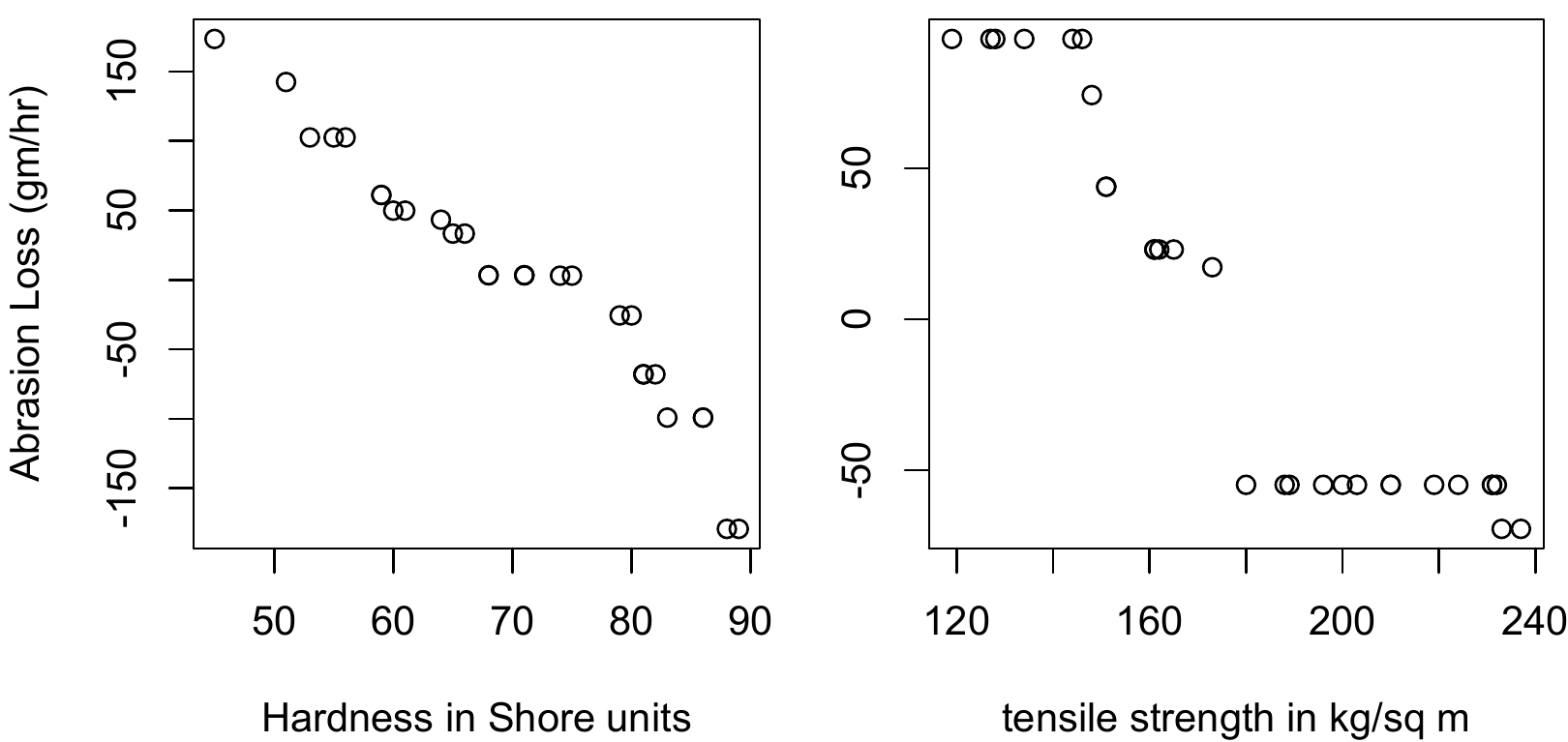}}
\caption{Centered components of the additive anti-tonic fit to the {\tt Rubber} data set.}
\label{fig:rub}
\end{figure}

\noindent {\bf Data example 3:}  To demonstrate the partial linear test, we use data from a study of predictors of blood plasma levels of the micronutrient beta carotene, in healthy subjects, as discussed by \citet{nierenberg89}.   Smoking status (current, former, never) and sex are categorical predictors whose effects on the response (log of blood plasma beta carotene) are determined to be significant at $\alpha = .05$.   If interest is in determining the effect of age of the subject, a linear relationship might be assumed; this fit is shown in the first plot of Fig.~\ref{fig:bpbc}, where the six lines correspond to the smoking/sex combinations.  The covariates must be included in the model because they are related to both the response and the predictor of interest.  To test whether the linear fit is appropriate, we use our test for linear versus partial linear model, which returns a $p$-value of $.047$.   The convex and concave fits have similar sums of squared residuals, with the concave fit having the smaller, so the concave fit represents the projection of the response vector onto the double cone.   However, the fact that the convex fit is almost as close to the data implies that neither is correct; perhaps the true function is concave at the left and convex at the right.

\begin{figure}
\centerline{$\hspace{0.7in}$\includegraphics[height=1.7in,width=6.2in]{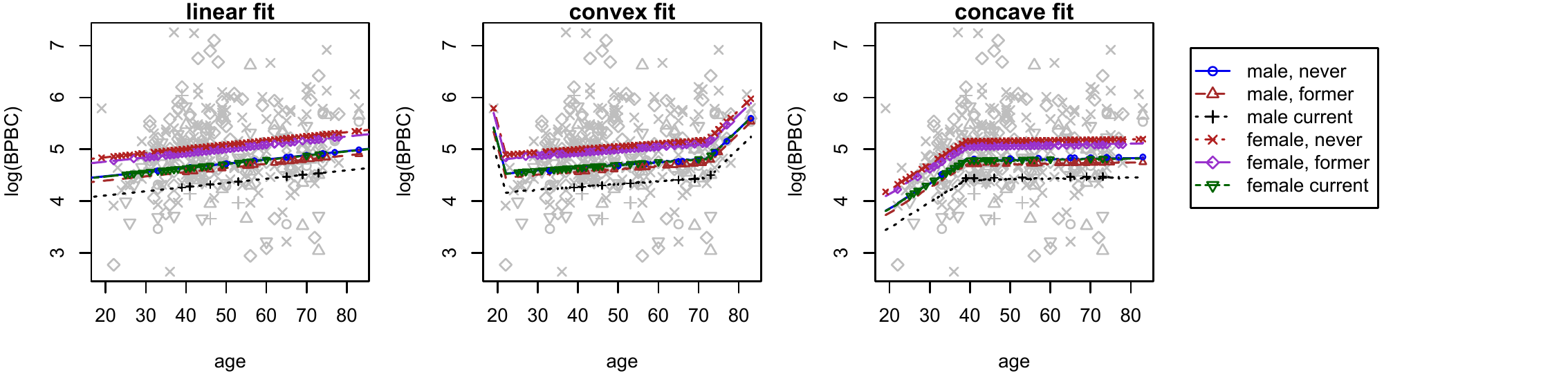}}
\caption{Log of blood plasma beta carotene (BPBC) is plotted against age, with plot character representing combinations of sex and smoking status. }
\label{fig:bpbc}
\end{figure}

\section{Discussion}
We have developed a test against a parametric regression function, where the alternative involves large-dimensional convex cones. The critical value of the test can be easily computed, via simulation, and the test is exact if we assume a known form of the error distribution. For a given parametric model,  a very general alternative is guaranteed to have power tending to one as the sample size increases, under mild conditions.  However, if additional {\em a priori} assumptions are available, these can be incorporated to boost the power in small to moderate-sized samples. For example, when testing against an additive linear function such as $\phi_0(x_1,x_2,x_3) = \beta_0+\beta_1x_1 + \beta_2x_2+\beta_3x_3$, we can use the ``fully convex'' model of Example~3, or if we feel confident that the additivity assumption is valid, we can use the octuple cone of Section~\ref{Additive}.   This power improvement was seen in model~2 simulations, in Fig.~\ref{LMpowerplot}.

The authors have provided {\tt R} and {\tt matlab} routines for the general method and for the specific examples.   In the {\tt R} package {\tt DoubleCone}, there are three functions.  The first, {\tt doubconetest} is the generic version; the user provides a constraint matrix that defines the cone $\I$ for which the null space of the constraint matrix is $\s$ and (A1) and (A2) hold.  The function provides a $p$-value for the test that the expected value of a vector is in the null space using the double-cone alternative.    The function {\tt agconst} performs a test of the null hypothesis that the expected value of {\tt y} is constant versus the alternative that it is monotone (increasing or decreasing) in each of the predictors, using double, quadruple, or octuple cones.   Finally, the function {\tt partlintest} performs a test of a linear model versus a partial linear model, using a double-cone alternative.   The user can test against a constant, linear, or quadratic function, while controlling for the effects of (optional) covariates.    The {\tt matlab} routine ({\tt http://www.stat.columbia.edu/$\sim$bodhi/Bodhi/Publications.html}) performs the test of Example 3.  

\appendix
\section{Appendix}
\subsection{Unbiasedness}\label{Unbiased}
We show that the power of the test for $\pmb{\theta}_0\in \I\cup\D$ is at least as large as the test size. In the following we give the proof of Theorem~\ref{unbiasedness} in the main paper.
\begin{thm}\label{Unbiasedness}
(Restatement of Theorem~\ref{unbiasedness}) Let $\mathbf{Y}_0 := \B s + \sigma \B \eps$, for $\B s \in \s$, where the components of $\B \eps$ are i.i.d.~$G$. Suppose further that $G$ is a symmetric (around 0) distribution. Choose any $\bt\in\Omega_I$, and let $\mathbf{Y}_1 :=\mathbf{Y}_0+\bt$. Then for any $a>0$,  
\begin{equation}\label{eq:Unbiasedness}
	\P\left(T(\mathbf{Y}_1)>a\right)\geq \P\left(T(\mathbf{Y}_0)>a\right).
\end{equation}
\end{thm}
Without loss of generality, we assume that $\B s = \B 0$ as the distribution of $T$ is invariant for any $\B s \in \s$, by Lemma~\ref{lem:Level}. 
To prove~\eqref{eq:Unbiasedness}, define $X_1 :=\|\Pi(\mathbf{Y}_0|\Omega_I)\|^2$ and $X_2 :=\|\Pi(-\mathbf{Y}_0|\Omega_I)\|^2$.  Then $X_1$ and $X_2$ have the same distribution as $G$ is a symmetric around 0, and $\|\Pi(\mathbf{Y}_0|\Omega_D)\|^2=\|\Pi(-\mathbf{Y}_0|\Omega_I)\|^2$.  In particular,
\[ \max\left\{\|\Pi(\mathbf{Y}_0|\Omega_I)\|^2,\|\Pi(\mathbf{Y}_0|\Omega_D)\|^2 \right\}=\max\left\{X_1,X_2\right\}=:T_0.
\]  Let $A$ be the event that $\{X_1\geq X_2\}$.  By symmetry $\P(A)=1/2$, and for any $a>0$,
\begin{eqnarray*}
\P(T_0\geq a) &=&\frac{1}{2} \left[\P(T_0\geq a|A)+\P(T_0\geq a|A^c)\right] \\
&=& \frac{1}{2} \left[ \P(X_1\geq a|A)+\P(X_2\geq a|A^c)\right].
\end{eqnarray*}
Let $\mathbf{Y}_2 :=\mathbf{Y}_0-\bt$ and define $W_1 :=\|\Pi(\mathbf{Y}_1|\Omega_I)\|^2$, $W_2 :=\|\Pi(-\mathbf{Y}_1|\Omega_I)\|^2$, $W_3 :=\|\Pi(\mathbf{Y}_2|\Omega_I)\|^2$, and $W_4 :=\|\Pi(-\mathbf{Y}_2|\Omega_I)\|^2$.  Then $W_1$ and $W_4$ are equal in distribution, as are $W_2$ and $W_3$.

\begin{lemma} $W_1\geq X_1$ and $W_4\geq X_2$.
\end{lemma}
\begin{proof} Using Lemma~\ref{projresid},
\begin{eqnarray*}
W_1 &=& \|\Pi(\mathbf{Y}_1|\Omega_I)\|^2 \\
&=& \|\Pi(\mathbf{Y}_1|\I)\|^2 -  \|\Pi(\mathbf{Y}_1|\s)\|^2 \\
&=&\|\mathbf{Y}_1-\Pi(\mathbf{Y}_1|\I^o)\|^2  -  \|\Pi(\mathbf{Y}_1|\s)\|^2\\
&=& \|\mathbf{Y}_1-\bt-[\Pi(\mathbf{Y}_1|\I^o)-\bt]\|^2 -  \|\Pi(\mathbf{Y}_1|\s)\|^2\\
&\geq&  \|\mathbf{Y}_0-\Pi(\mathbf{Y}_0|\I^o)\|^2  -  \|\Pi(\mathbf{Y}_1|\s)\|^2 \\
&=&  \|\Pi(\mathbf{Y}_0|\I)\|^2  -  \|\Pi(\mathbf{Y}_1|\s)\|^2 \\
&=&  \|\Pi(\mathbf{Y}_0|\Omega_I)\|^2 =X_1,   \\
\end{eqnarray*}
where the last equality uses $\Pi(\mathbf{Y}_1|\s)=\Pi(\mathbf{Y}_0|\s)$.  The inequality holds because $\Pi(\mathbf{Y}_1|\I^o)-\bt\in\I^o$, by (A2).  The proof of $W_4\geq X_2$ is similar.
\end{proof}
{\em Proof of Theorem~\ref{Unbiasedness}:}
Let $S=\max(W_1,W_2)$, which is equal in distribution to $\max(W_3,W_4)$.  For any $a>0$,
\begin{eqnarray*}
 \P(S>a)&=&\frac{1}{2}\left[ \P(S>a|A)+\P(S>a|A^c)\right]  \\
& \geq& \frac{1}{2}\left[ \P(W_1>a|A)+\P(W_4>a|A^c)\right]  \\
& \geq& \frac{1}{2}\left[ \P(X_1>a|A)+\P(X_2>a|A^c)\right]  \\
&=& \P(T_0>a).
\end{eqnarray*}

Finally, we note that $T(\mathbf{Y}_1) =S/SSE_0$ and $T(\mathbf{Y}_0) =T_0/SSE_0$,
where $SSE_0$ is the sum of squared residuals of the projection of either $\mathbf{Y}_1$ or $\mathbf{Y}_0$ onto $\s$,  because $\bt$ is orthogonal to $\s$.  \qed

\subsection{Some intuition under which the power goes to~1}\label{PowerTo1}
The following result shows that if both projections $\bt_I$ and $\bt_D$ belong to $\s$ then $\bt_0$ must itself lie in $\s$, if $\I$ is a ``large'' cone. This motivates the fact that if $\bt_0 \notin \s$, then both $\bt_I$ and $\bt_D$ cannot be very close to $\bt_S$, and~\eqref{eq:PowerC} might hold. The largeness of $\I$ can be represented through the following condition. Suppose that any $\pmb \xi \in \R^n$ can be expressed as 
\begin{eqnarray}\label{eq:LargeI}
\pmb \xi = \pmb \xi_I + \pmb \xi_D
\end{eqnarray} 
for some $\pmb \xi_I \in \I$ and $\pmb \xi_D \in \D$. This condition holds for $\I=\{\bt:\pmb{A}\bt\geq\pmb{0}\}$ where $\pmb{A}$ is {\em irreducible}, and $\s$ is the null row space of $\pmb{A}$.   A constraint matrix is irreducible as defined by \citet{meyer99} if the constraints defined by the rows are in a sense non-redundant.  Then bases for $\s$ and $\Omega_I$ together span $\R^n$, so any $\bt_0 \in \R^n$ can be written as the sum of vectors in $\s$, $\Omega_I$, and $\Omega_D$ simply by writing $\bt_0$ as a linear combination of these basis vectors, and gathering terms with negative coefficients to be included in the $\Omega_D$ component.

\begin{lemma}\label{lem:UniCst}
If~\eqref{eq:LargeI} holds then $\pmb \theta_0 \in \s$ if and only if $\pmb \theta_I \in \s$ and $\pmb \theta_D \in \s$.
\end{lemma}
\begin{proof}
Suppose that $\pmb \theta_0 \in \s$. Then $\pmb \theta_0 \in \I$ and thus $\pmb \theta_I = \pmb \theta_0$. Similarly, $\pmb \theta_0 \in \D$ and $\pmb \theta_D = \pmb \theta_0$. Hence,  $\pmb \theta_0 = \pmb \theta_I = \pmb \theta_D \in \s$.

Suppose now that $\pmb \theta_I, \pmb \theta_D \in \s$.   By~\eqref{eq:ProjC}, $\langle \bt_0-\bt_I,\pmb s\rangle =\langle \bt_0-\bt_D,\pmb s\rangle =0$ for any $\pmb s\in{\s}$; this implies that $\langle\bt_I,\pmb{s}\rangle=\langle\bt_D,\pmb{s}\rangle$ for all $\pmb s\in{\s}$, so $\bt_I=\bt_D=:\pmb\gamma$.   From~(\ref{eq:InProj}) applied to ${\cal D}$ and ${\cal I}$ it follows that
$$\langle \pmb \theta_{0} - \pmb \gamma, \pmb \xi_I + \pmb \xi_D \rangle \le 0,$$ for all $\pmb \xi_I \in \I$ and $\pmb \xi_D \in \D$. As any $\pmb \xi \in \R^n$ can be expressed as $\pmb \xi = \pmb \xi_I + \pmb \xi_D$ for $\pmb \xi_I \in \I$ and $\pmb \xi_D \in \D$, the above display yields $\langle \pmb \theta_{0} - \pmb \gamma, \pmb \xi \rangle \le 0$ for all $\pmb \xi \in \R^n$. Taking $\pmb \xi$ and $-\pmb \xi$ in the above display we get that $\langle \pmb \theta_{0} - \pmb \gamma, \pmb \xi \rangle = 0$ for all $\pmb \xi \in \R^n$, which implies that $\pmb \theta_{0} = \pmb \gamma$, thereby proving the result.
\end{proof}

\subsection{Testing against a constant function}\label{TestCons}\label{TestCnst}
The traditional $F$-test for the parametric least-squares regression model has the null hypothesis that none of the predictors is (linearly) related to the response. For an $n\times p$ full-rank design matrix, the $F$ statistic has null distribution $F(p-1,n-p)$. To test against the constant function when the relationship of the response with the predictors is unspecified, we can turn to our cone alternatives.   

Consider model~(\ref{eq:RegMdl}) where $\phi_0$ is the unknown true regression function and ${\cal X} =\{\pmb{x}_1, \pmb{x}_2,\ldots, \pmb{x}_n\} \subset \R^d$ is the set of predictor values.  We can assume without loss of generality that there are no duplicate $\pmb{x}$ values; otherwise we average the response values at each distinct $\pmb{x}$, and do weighted regression.  Interest is in $H_0:\phi_0\equiv c$ for some unknown scalar $c\in\R$, against a general alternative. A cone that contains the one-dimensional null space is defined for multiple isotonic regression.  

We start with some definitions. A partial ordering on $\mathcal{X}$ may be defined as $\pmb{x}_i\preceq \pmb{x}_j$ if $\pmb{x}_i\leq \pmb{x}_j$ holds coordinate-wise.  Two points $\pmb{x}_i$ and $\pmb{x}_j$ in  ${\cal X}$  are comparable if either $\pmb{x}_i\preceq\pmb{x}_j$ or $\pmb{x}_j\preceq\pmb{x}_i$.   Partial orderings are reflexive, anti-symmetric, and transitive, but differ from complete orderings in that pairs of points are not required to be comparable. 
A function $\phi_0:{\cal X}\rightarrow\R$ is isotonic with respect to the partial ordering if $\phi_0(\pmb{x}_i)\leq \phi_0(\pmb{x}_j)$ whenever $\pmb{x}_i\preceq \pmb{x}_j$.   If $\bt\in\R^n$ is defined as $\theta_i= \phi_0(\pmb{x}_i)$, we can consider the set $\I$ of $\bt\in\R^n$ such that $\theta_i\leq \theta_j$ whenever $\pmb{x}_i\preceq \pmb{x}_j$.  The set $\I$ is a convex cone in $\R^n$, and a constraint matrix $\pmb{A}$ can be found so that $\I=\{\bt:\pmb{A}\bt\geq\pmb 0\}$.

Assumption (A1) holds if ${\cal X}$ is {\em connected}; i.e., for all proper subsets ${\cal X}_0\subset{\cal X}$, each point in ${\cal X}_0$ is comparable with at least one point not in ${\cal X}_0$.   If ${\cal X}$ is not connected, it can be broken down into smaller connected subsets, and the null space of $\pmb{A}$ consists of vectors that are constant over the subsets. We have the following lemma.  
\begin{lemma} The null space $\s$ of the constraint matrix $\pmb{A}$ associated with isotonic regression on the set ${\cal X}$ is spanned by the constant vectors if and only if the set ${\cal X}$ is connected.
\label{vconst}
\end{lemma}
\begin{proof}
Let $\pmb{A}$ be the constraint matrix associated with isotonic regression on the set ${\cal X}$ such that there is a row for every comparable pair  (i.e., before reducing).  Let ${\cal C}$ be the null space of $\pmb{A}$ (it is easy to see that the null space of the reduced constraint matrix is also $\s$).  
First, suppose $\pmb{\theta}\in \s$ and $\theta_a\neq \theta_b$ for some $1\leq a,b\leq n$.  Let ${\cal X}_a \subset {\cal X}$ be the set of points that are comparable to $\pmb{x}_a$, and let ${\cal X}_b  \subset {\cal X}$ be the set of points that are comparable to $\pmb{x}_b$.  If $\pmb{x}_j$ is in the intersection, then there is row of $\pmb{A}$ where the $j$-th element is $-1$ and the $a$-th element is $+1$ (or vice-versa), as well as a row where the $j$-th element is $-1$ and the $b$-th element is $+1$ (or vice-versa), so that $\theta_a=\theta_j$ and $\theta_b=\theta_j$.  Therefore if $\theta_a\neq \theta_b$, ${\cal X}_a\cap{\cal X}_b$ is empty and ${\cal X}$ is not connected.
Second, suppose ${\cal X}$ is connected and $\pmb{\theta}\in \s$.  For any $a\neq b$, $1\leq a,b\leq n$, define ${\cal X}_a$ and ${\cal X}_b$ as above; then because of connectedness there must be an $\pmb{x}_j$ in the intersection, and hence $\theta_a=\theta_b$. Thus only constant vectors are in $\s$.  
\end{proof}

The classical isotonic regression with a partial ordering can be formulated using upper and lower sets.   The set $U\subseteq{\cal X}$ is an upper set with respect to $\preceq$ if $\pmb{x}_1\in U$, $\pmb{x}_2\in {\cal X}$, and $\pmb{x}_1\preceq \pmb{x}_2$ imply that $\pmb{x}_2\in U$.     Similarly, the set $L\subseteq{\cal X}$ is a lower set with respect to $\preceq$ if $\pmb{x}_2\in L$, $\pmb{x}_1\in {\cal X}$, and $\pmb{x}_1\preceq \pmb{x}_2$ imply that $\pmb{x}_1\in L$.    

The following results can be found in \citet{gebhardt70}, \citet{barlow72}, and \citet{dykstra81}.  Let ${\cal U}$ and ${\cal L}$ be the collections of upper and lower sets in ${\cal X}$, respectively. The isotonic regression estimator at $\pmb{x}\in{\cal X}$ has the closed form
\[\hat{\phi}_0(\pmb{x}) = \max_{U\in{\cal U}: \pmb{x} \in U} \min_{L\in{\cal L}:\pmb{x}\in L} Av_{\mathbf Y}(L\cap U),\]
where $Av_{\mathbf{Y}}(S)$ is the average of all the $Y$ values for which the predictor values are in $S\subseteq {\cal X}$.  Further, if $\bt\in\I$, then for any upper set $U$ we have
\begin{equation}\label{eq:upperlower} 
Av_{\bt}(U)\geq Av_{\bt}(U^c).
\end{equation}
Similarly, if $L$ is a lower set and $\bt\in\I$, $Av_{\bt}(L)\leq Av_{\bt}(L^c)$, as the complement of an upper set is a lower set and vice-versa.    Finally, any $\bt\in\I$ can be written as a linear combination of indicator functions for upper sets, with non-negative coefficients, plus a constant vector.   

We could define a double cone where $\D=\{\bt:\pmb{A}\bt\leq\pmb 0\}$ for antitonic $\bt$.  To show the condition (A2) holds, we first show that if $\bt\in\Omega_I=\I\cap\s^{\perp}$, then the projection of $\bt$ on $\D$ is the origin; hence $\Omega_I\subseteq\D^o$.    Let $\bt_{U}$ be the ``centered'' indictor vector for an upper set $U$.  That is, $\bt_{Ui}=a$ if $\pmb{x}_i\in U$, $\bt_{Ui}=b$ if $\pmb{x}_i\notin U$, and $\sum_{i=1}^n\bt_{Ui}=0$.   Then $\bt_U\in\Omega_I$, and $\langle \bt_U, \bt\rangle\geq 0$ for any $\bt\in\I$ by (\ref{eq:upperlower}).  For, 
\begin{eqnarray*}
 \langle \bt_U, \bt\rangle &=&a\sum_{i:x_i\in U^c} \theta_i + b\sum_{i:x_i\in U} \theta_i \\
 &=& (n-n_u)aAv_{\theta}(U^c)+n_ubAv_{\theta}(U) \\
 &\geq& \{(n-n_u)a+n_ub\}Av_{\theta}(U^c) =0,
\end{eqnarray*}
where $n_u$ is the number of elements in $U$.

  Similarly, $\langle \bt_U, \br\rangle\leq 0$ for any $\br\in\D$.   We can write any $\bt\in\Omega_I$ as a linear combination of centered upper set indicator vectors with non-negative coefficients, so that $\langle \bt, \br\rangle\leq 0$ for any $\bt\in\Omega_I$ and $\br\in\D$.
Then for any $\bt\in\Omega_I$ and $\br\in\D$,
\[\| \bt-\br\|^2 = \|\bt\|^2 +\|\br\|^2 -2\langle \bt, \br\rangle \geq \|\bt\|^2,
\]
hence the projection of $\bt\in\Omega_I$ onto $\D$ is the origin.

The double cone for multiple isotonic regression is unsatisfactory because if one of the predictors reverses sign, the value of the statistic~(\ref{eq:TestS}) (for testing against a constant function) also changes.   For two predictors, it is more appropriate to define a quadruple cone.   Define:
\begin{itemize}
\item $\I_1$: the cone defined by the partial ordering: $\pmb{x}_i\preceq \pmb{x}_j$ if $x_{1i}\leq x_{1j}$  and $x_{2i} \leq x_{2j}$;
\item $\I_2$: the cone defined by the partial ordering: $\pmb{x}_i\preceq \pmb{x}_j$ if $x_{1i}\leq x_{1j}$  and $x_{2i}\geq x_{2j}$;
\item $\I_3$: the cone defined by the partial ordering: $\pmb{x}_i\preceq \pmb{x}_j$ if $x_{1i}\geq x_{1j}$  and $x_{2i}\leq x_{2j}$; and
\item $\I_4$: the cone defined by the partial ordering: $\pmb{x}_i\preceq \pmb{x}_j$ if $x_{1i}\geq x_{1j}$  and $x_{2i}\geq x_{2j}$.
\end{itemize}
The cones $\I_1$ and $\I_4$ form a double cone as do $\I_2$ and $\I_3$.   
If ${\cal X}$ connected, the one-dimensional space $\s$ of constant vectors is the largest linear space in any of the four cones and (A1) is met.   

Let $\hat \bt_j$ be the projection of $\pmb{Y}$ onto $\I_j$, and $SSE_j=\|\pmb{Y} - \hat \bt_j\|^2$ for $j=1,2,3,4$, while $SSE_0=\sum_{i=1}^n(Y_i-\bar{Y})^2$.   Define $T_j=(SSE_0-SSE_j)/SSE_0$, for $j=1,2,3,4$, 
and $T= \max\{T_1,T_2,T_3,T_4\}$.
The distribution of $T$ is invariant to translations in $\s$; this can be proved using the same technique as Lemma~\ref{lem:DistFreeH_0}.  Therefore the null distribution can be simulated up to any desired precision for known $G$.  To show that the test is unbiased for $\bt_0$ in the quadruple cone $\I_1\cup\I_2\cup\I_3\cup\I_4$, we note that condition (A2) becomes $\Omega_1\subseteq\I_4^o$ and $\Omega_2\subseteq\I_3^o$, and vice-versa; the results of Section~\ref{Unbiased} hold for the quadruple cone. For three predictors we need an octuple cone, which is comprised of four double-cones and similar results follow.   

The results of Section~\ref{AsymPower} depend on the consistency of the multiple isotonic regression estimator. Although \citet{hanson73} proved point-wise and uniform consistency (on compacts in the interior of the support of the covariates) for the projection estimator in the bivariate case (also see~\citet{makowski77}), result~(\ref{eq:ConsEst}) for the general case of multiple isotonic regression is still an open problem.

\subsubsection{Simulation study}
We consider the test against a constant function using the quadruple cone alternative of Section~\ref{TestCnst}, and model $Y_i=\phi_0(x_{1i},x_{2i})+\epsilon_i$, $i=1,\ldots,100$, where for each simulated data set, the $(x_1,x_2)$ values are generated uniformly in the unit square.  The power is compared with the standard $F$-test where the alternative model is $Y_i=\beta_0+\beta_1x_{1i}+\beta_2x_{2i}+\beta_3x_{1i}x_{2i}+\epsilon_i$, the FH test with known variance, and $S_1$ and $S_2$.  In the first plot of Figure~\ref{quadpowerplot2}, power is shown when the true regression function is $\phi_0(x_1,x_2)=2ax_1x_2$.  Here the assumptions for the parametric $F$-test are correct, so the $F$-test has the highest power, although the power for the proposed test is only slightly smaller.   In the second plot, the regression function is quadratic in $x_1$, with vertex in the center of the design points.   The $F$-test fails to reject the constant model because of the non-linear alternative, but the true regression function is also far from the quadruple-cone alternative.   However, the proposed test has power comparable to the FH test with known variance.   In the final plot, the regression function is constant on $[0,2/3]\times[0,2/3]$, and increasing in both predictors beyond 2/3.  The proposed test has the best power compared with the alternatives.  The $S_1$ and $S_2$ tests do not perform well for the test against a constant function in two dimensions, compared with other testing situations.

\begin{figure}
\centerline{\includegraphics[height=2.1in, width=6in]{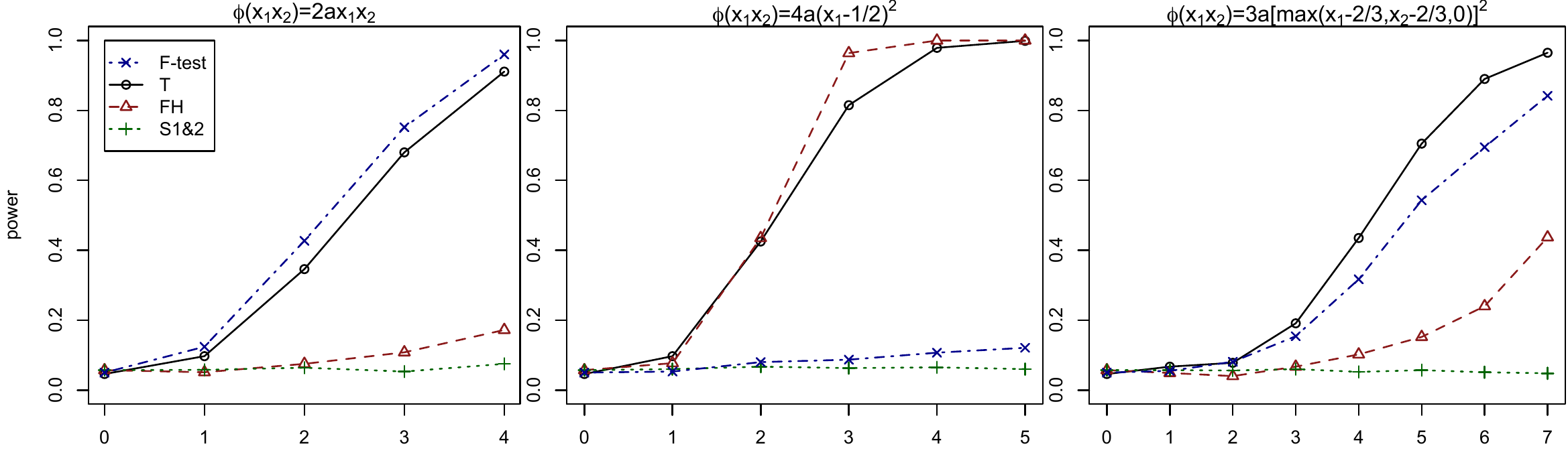} }
\label{quadpowerplot2}
\caption{Power function for test against the constant model, with $n=100$ observations, predictors generated uniformly in $[0,1]^2$ and $\sigma^2=1$.}
\end{figure}

%%%%%%%%%%%%%%%%%%%%%%%%%%%%%%%%%%%%%%%%%%%%%%
\subsection{Proofs of Lemmas and Theorems}\label{Proofs}
\noindent {\bf Proof of Theorem~\ref{thm:CutOff}}
To study the distribution of $T(\hat {\B \eps}) \sim D_n$ and relate it to that of $T({\B \eps}) \sim H_n$, we consider the following quantile coupling: Let $Z_1,\ldots$ be a sequence of i.i.d.~Uniform(0,1) random variables. Define $\eps_j = G^{-1}(Z_j)$ and let $\hat \eps_{nj} = \hat G_n^{-1}(Z_j)$, for $j=1,\ldots, n$; $n \ge 1$. Observe that $T(\eps_1,\ldots, \eps_n) \sim H_n$ and the conditional distribution of $T(\hat {\eps}_{n1},\ldots,\hat {\eps}_{nn})$, given the data, is $D_n$. For notational simplicity, for the rest of the proof, we will denote by $ {\B \eps} := (\eps_1,\ldots, \eps_n)$ and by $\hat {\B \eps} := (\hat {\eps}_{n1},\ldots,\hat {\eps}_{nn})$. Note that,
\begin{eqnarray}\label{eq:ProjLip}
	& & \|\Pi(\B \eps|\s) - \Pi(\B \eps|\I)\| \notag \\
& = & \|\{\Pi(\B \eps|\s) - \Pi(\hat{\B \eps}|\s)\} + \{\Pi(\hat{\B \eps}|\s) - \Pi(\hat{\B \eps}|\I)\} + \{\Pi(\hat{\B \eps}|\I) - \Pi(\B \eps|\I)\}\| \notag \\
	& \le & \|\Pi(\B \eps|\s) - \Pi(\hat{\B \eps}|\s)\| + \|\Pi(\hat{\B \eps}|\s) - \Pi(\hat{\B \eps}|\I)\| + \|\Pi(\hat{\B \eps}|\I) - \Pi(\B \eps|\I)\| \notag \\
	& \le & 2 \|\B \eps - \hat{\B \eps}\| + \|\Pi(\hat{\B \eps}|\s) - \Pi(\hat{\B \eps}|\I)\|
\end{eqnarray}
where we have used the fact that the projection operator is 1-Lipschitz. To simplify notation, let $U :=  \|\Pi(\B \eps|\s) - \Pi(\B \eps|\I)\|$, $V := \|\Pi(\B \eps|\s) - \Pi(\B \eps|\D)\|$ and $W := \|\B \eps - \Pi(\B \eps|\s)\|$. Also, let $\hat U :=  \|\Pi(\hat {\B \eps}|\s) - \Pi(\hat {\B \eps}|\I)\|$, $\hat V := \|\Pi(\hat {\B \eps}|\s) - \Pi(\hat {\B \eps}|\D)\|$ and $\hat W := \|\hat {\B \eps} - \Pi(\hat {\B \eps}|\s)\|$. 
Thus, using a similar argument as in~\eqref{eq:ProjLip} we obtain,
$$\max \left\{ |U-\hat U|,|V - \hat V|,|W - \hat W| \right\} \le 2 \|\B \eps - \hat{\B \eps}\|.$$ Therefore,
\begin{eqnarray}\label{eq:BoundTeps}
|{T^{1/2}(\B \eps)} - {T^{1/2}(\hat{\B \eps})}|& = & \left|\frac{\max  \left\{U, V \right\}}{W} - \frac{\max \{\hat U, \hat V\}}{\hat W} \right| \notag \\
& \le & \frac{|\max \{U, V \} - \max \{\hat U, \hat V\}|}{W}  +  \max \{\hat U, \hat V\} \left| \frac{1}{W} - \frac{1}{\hat W}\right| \notag \\ 
& \le &\frac{\max \{|U - \hat U|, |V - \hat V|\}}{W} +  \frac{ \max \{\hat U, \hat V \}}{\hat W} \frac{| \hat W - W|}{W} \notag \\ 
& \le &  2 \frac{\|\B \eps - \hat{\B \eps}\|}{W} + \frac{| \hat W - W|}{W} \le 4 \frac{\|\B \eps - \hat{\B \eps}\|}{W} 
\end{eqnarray}
as $\max \{\hat U, \hat V\}/{\hat W} \le 1$. For two probability measure $\mu$ and $\nu$ on $\R$, let $d_p(\mu,\nu)$ denote the $p$-Wasserstein distance between $\mu$ and $\nu$, i.e., $$[d_p(\mu,\nu)]^p := \inf_J \{\E |S - T|^p: S \sim \mu, T \sim \nu\},$$ where the infimum is taken over all joint distributions $J$ with marginals $\mu,\nu$. 
Now, 
\begin{eqnarray}\label{eq:error}
\E \left( \frac{1}{n} \| \B \eps - \hat{\B \eps}\|^2 \right) = \frac{1}{n} \sum_{j=1}^n \E (\eps_i - \hat \eps_{nj})^2 = \int_{0}^1 |G_n^{-1}(t) - G^{-1}(t)|^2 dt = d_2(G_n,G),\;\;\;\;
\end{eqnarray}
where the last equality follows from~\citet[Theorem 2, page 64]{SW86}. Further, by~\citet[Theorem 1, page 63]{SW86}, and two conditions on $\hat G_n$ stated in the theorem, it follows that $d_2(G_n,G) \to 0$ a.s. 
Therefore, 
\begin{eqnarray*}
	d_1(H_n,D_n) & \le & \E |T(\hat {\B \eps}) - T({\B \eps})| \;\;	\le \;\; 2 \E |{T^{1/2}(\hat {\B \eps})} - T^{1/2}({\B \eps})|  \\
	& \le & 8 \E \left( \frac{\|\B \eps - \hat{\B \eps}\|}{W}\right) \;\; \le \;\; 8 \sqrt{\E \left( \frac{n}{W^2} \right) \E \left( \frac{1}{n} \|\B \eps - \hat{\B \eps}\|^2\right)} \to 0 \quad \mbox{a.s.},
\end{eqnarray*}
where we have used~\eqref{eq:BoundTeps} and~\eqref{eq:error}, $ \E ( n W^{-2} )$ is uniformly bounded, and the fact that $d_2(G_n,G) \to 0$ a.s. The result now follows from the fact that $d_L(H_n,D_n) \le \sqrt{d_1(H_n,D_n)}$; see \citet[pages 31--33]{Huber81}. \newline \qed

{\bf Proof of Lemma~\ref{lem:conv}:}  Let $\pmb\rho := \bt_0-\bt_I$.  Then $\pmb\rho\perp\bt_I$.  Let $\pmb{Z}=\mathbf{Y}-\pmb{\rho}$, and note that $\pmb{Z} =\bt_I+\pmb{\epsilon}$. If $\check{\bt}_I$ is the projection of $\pmb{Z}$ onto $\I$, then 
\begin{equation}\label{eq:FactZ}
\| \check{\bt}_I-\bt_I \|^2=o_p(n) \ \mbox{ and } \ \langle \pmb{Z}-\check{\bt}_I,\bt_I \rangle =o_p(n).
\end{equation}
The first follows from assumption~\eqref{eq:ConsEst} and the latter holds because
\[ 0 \leq -\langle \pmb{Z}-\check{\bt}_I,\bt_I\rangle = \langle \pmb{Z}-\check{\bt}_I, \check{\bt}_I - \bt_I\rangle  \leq \|\pmb{Z}-\check{\bt}_I\| \|\check\bt_I -\bt_I\| = o_p(n),
\] as $\|\pmb{Z}-\check{\bt}_I\|^2/n\leq \|\pmb{Z} - \bt_I \|^2/n=O_p(1)$, where we have used the characterization of projection on a closed convex cone~\eqref{eq:InProj}. Starting with
\[ \|\check{\bt}_I -\bt_I\|^2 = \|\check{\bt}_I -\hat{\bt}_I \|^2 +\|\hat{\bt}_I -\bt_I\|^2 + 2\langle \check{\bt}_I -\hat{\bt}_I, \hat{\bt}_I -\bt_I\rangle,
\]
we rearrange to get 
{\small \begin{eqnarray}\label{eq:Simp1}
 \|\check{\bt}_I -\bt_I\|^2 - \|\hat{\bt}_I -\bt_I\|^2  &=&  \|\check{\bt}_I -\hat{\bt}_I \|^2+2 \langle \check{\bt}_I -\hat{\bt}_I, \hat{\bt}_I -\bt_I\rangle \nonumber \\
 &=& \|\check{\bt}_I -\hat{\bt}_I \|^2 + 2\langle \check{\bt}_I - \mathbf{Y} , \hat{\bt}_I - \bt_I\rangle  +2 \langle \mathbf{Y} - \hat{\bt}_I, \hat{\bt}_I -\bt_I\rangle. \qquad  \qquad
\end{eqnarray}}
As $\mathbf{Y}= \pmb{Z}+\pmb{\rho}$ and $\pmb\rho\perp\bt_I$, we have
\begin{eqnarray*}
\langle \mathbf{Y} -\check{\bt}_I, \hat{\bt}_I - \bt_I \rangle & = & \langle \mathbf{Y} -\check{\bt}_I, \hat{\bt}_I \rangle - \langle \mathbf{Y} -\check{\bt}_I  ,\bt_I\rangle \\
& = & [\langle \pmb{Z} -\check{\bt}_I, \hat{\bt}_I \rangle +\langle \pmb{\rho}, \hat{\bt}_I \rangle] - \langle \pmb{Z}-\check{\bt}_I,\bt_I\rangle.
\end{eqnarray*}
Also, $\langle \mathbf{Y} - \hat{\bt}_I, \hat{\bt}_I -\bt_I\rangle = - \langle \mathbf{Y} - \hat{\bt}_I, \bt_I\rangle$. Thus,~\eqref{eq:Simp1} equals
 \begin{eqnarray*}
\|\check{\bt}_I -\hat{\bt}_I \|^2 - 2 \langle \mathbf{Y} - \hat{\bt}_I,\bt_I\rangle -2 \langle \pmb{Z} -\check{\bt}_I, \hat{\bt}_I \rangle -2\langle \pmb{\rho}, \hat{\bt}_I \rangle  + 2\langle \pmb{Z}-\check{\bt}_I,\bt_I\rangle.
 \end{eqnarray*}
The first four terms in the above expression are positive (with their signs), and the last is $o_p(n)$ by~\eqref{eq:FactZ}. Therefore, $$ \|\hat{\bt}_I -\bt_I\|^2  \le \|\check{\bt}_I -\bt_I\|^2 + o_p(n),$$ which when combined with~\eqref{eq:FactZ} gives the desired result. The proof of the other part is completely analogous. \newline \qed

\noindent {\bf Proof of Theorem~\ref{PowerCond}:} To show that \eqref{eq:PowerC} holds let us define the convex projection of $\phi_0$ with respect to the measure $\mu$ as $$\phi_I = \argmin_{\phi \mbox{ convex }} \int_{[0,1]^d} (\phi_0(\mathbf{x}) - \phi(\mathbf{x}))^2 d \mu(\mathbf{x}),$$ where the minimization is over all convex functions from $[0,1]^d$ to $\R$. The existence and uniqueness (a.s.)~of $\phi_I$ follows from the fact that $\h := L_2([0,1]^d, \mu)$ is a Hilbert space with the inner product $\langle f,g\rangle_\h = \int_{[0,1]^d} f(\mathbf{x}) g(\mathbf{x}) d \mu(\mathbf{x})$, and the space $\tilde \I$ of all convex functions in $\h$ is a closed convex set in $\h$. Let $\tilde \D$ of the space of all concave functions in $\h$. We can similarly define the non-increasing projection $\phi_D$ of $\phi_0$.

Next we show that if $\phi_0$ is not affine a.e.~$\mu$ then either $\phi_I$ is not affine a.e.~$\mu$ or $\phi_D$ is not affine a.e.~$\mu$. Suppose not, i.e., suppose that $\phi_I$ and $\phi_D$ are both affine a.e.~$\mu$. Then by~\eqref{eq:ProjC}, $\langle \phi_0 - \phi_I, f \rangle_\h = \langle \phi_0-\phi_D, f\rangle_\h =0$ for any affine $f$. This implies that $\langle\phi_I,f\rangle_\h=\langle\phi_D,f\rangle_\h$ for all affine $f$, so $\phi_I=\phi_D=: \phi_S$ a.e., where $\phi_S$ is affine. Note that $\phi_S$ is indeed the projection of $\phi_0$ onto the space of all affine functions in $\h$. From~(\ref{eq:InProj}) applied to ${\tilde \D}$ and ${\tilde \I}$ it follows that
$$\langle \phi_{0} - \phi_S, f_I + f_D \rangle_\h \le 0,$$ for all $f_I \in \tilde \I$ and $f_D \in \tilde \D$. As any $f \in \h$ that is twice continuously differentiable can be expressed as $f = f_I + f_D$ for $f_I \in \tilde \I$ and $f_D \in \tilde \D$, the above display yields $\langle \phi_{0} - \phi_S, f \rangle_\h \le 0$ for all $f$ that is twice continuously differentiable. Taking $f$ and $-f$ in the last inequality we get that $\langle \phi_{0} - \phi_S, f \rangle_\h = 0$ for all $f$ that is twice continuously differentiable, which implies that $\phi_{0} = \phi_S$ a.e.~$\mu$, giving rise to a contradiction.

We can also show that $$\lim_{n \rightarrow \infty} \frac{1}{n} \|{\pmb  \theta}_I - {\pmb  \theta}_S\|^2 =  \int_{[0,1]^d} (\phi_I(\mathbf{x}) - \phi_S(\mathbf{x}))^2 d \mu(\mathbf{x}) > 0.$$ Similarly, ${n}^{-1} \|{\pmb  \theta}_D - {\pmb  \theta}_S\|^2$ also converges to a positive number. This proves~\eqref{eq:PowerC}. \newline
\qed

\noindent {\bf Proof of Theorem~\ref{thm:Null}:} Observe that $T_I = \|\hat{\pmb \theta}_S - \hat{\pmb \theta}_I\|^2/\|\mathbf{Y} - \hat{\pmb \theta}_S\|^2$.  Letting $\mathbf{e}_n := (1,1,\ldots, 1)^\top \in \R^n$, and noting that $\hat{\pmb \theta}_S = \bar Y \mathbf{e}_n$, where $\bar Y = \sum_{i=1}^n Y_i/n$, we have
\begin{equation*}
\frac{1}{n}\|\mathbf{Y} - \hat{\pmb \theta}_S\|^2  \rightarrow_p \sigma^2 > 0.
\end{equation*} 
Now,
\begin{eqnarray*}
	\|\hat{\pmb \theta}_S - \hat{\pmb \theta}_I\|^2 & = & \|\hat{\pmb \theta}_S - \pmb \theta_0 + \pmb \theta_0  - \hat{\pmb \theta}_I\|^2 \\
	& \le & 2 \|\hat{\pmb \theta}_S - \pmb \theta_0 \|^2 + 2 \|\hat{\pmb  \theta}_I - \pmb \theta_0 \|^2 \\
	& = & O_p(1) + O_p(\log n),
\end{eqnarray*}
where we have used the facts that $\|\hat{\pmb \theta}_S - \pmb \theta_0 \|^2 = O_p(1)$ and $ \|\hat{\pmb  \theta}_I - \pmb \theta_0\|^2 = O_p(\log n)$ (see \citet{CGS13}; also see~\citet{MW00} and~\citet{Zhang02}). A similar result can be obtained for $T_D$ to arrive at the conclusion $T = O_p(\log(n)/n)$. \newline
\qed

\noindent {\bf Proof of Theorem~\ref{thm:Misspec}:} We will verify \eqref{eq:ConsEst} and \eqref{eq:PowerC} and apply Theorem~\ref{thm:h1part} to obtain the desired result. First observe that \eqref{eq:ConsEst} follows immediately from the results in \citet{CGS13}; also see \citet{Zhang02}. To show that \eqref{eq:PowerC} holds let us define the non-decreasing projection of $\phi_0$ with respect to the measure $F$ as $$\phi_I = \argmin_{\phi \uparrow} \int_0^1 (\phi_0(x) - \phi(x))^2 d F(x),$$ where the minimization is over all non-decreasing functions. We can similarly define the non-increasing projection $\phi_D$ of $\phi_0$.

As $\phi_0$ is not a constant a.e.~$F$, it can be shown by a very similar argument as in Lemma~\ref{lem:UniCst} that either $\phi_I$ is not a constant a.e.~$F$ or $\phi_D$ is not a constant a.e.~$F$ (note that here we use the fact that a function of bounded variation on an interval can be expressed as the difference of two monotone functions). Without loss of generality let us assume that $\phi_I$ is not a constant a.e.~$F$. Therefore, $$\lim_{n \rightarrow \infty} \frac{1}{n} \|{\pmb  \theta}_I - {\pmb  \theta}_S\|^2 =  \int_0^1 (\phi_I(x) - c_0)^2 d F(x) > 0,$$ which proves \eqref{eq:PowerC}, where $c_0 = \argmin_{c \in \R} \int_0^1 (\phi_0(x) - c)^2 d F(x)$. \newline \qed

\noindent {\bf Proof of Theorem~\ref{thm:MultCvxCons}:}
The theorem follows from known metric entropy results on the class of uniformly bounded convex functions that are uniformly Lipschitz in conjunction with known results on consistency of least squares estimators; see Theorem 4.8 of \citet{vdG00}. We give the details below. 

The notion of covering numbers will be used in the
sequel. For $\epsilon > 0$ and a subset $\G$ of functions, the
$\epsilon$-covering number of $\G$ under the metric $\ell$, denoted by $N(S, \epsilon;\ell)$, is defined as the smallest number of closed balls of radius $\epsilon$ whose union contains $\G$.

Fix any $B>0$ and $L> L_0$. Recall the definition of the class $ \tilde \I_L$, given in~\eqref{eq:C_L}. We define the class of uniformly bounded convex functions that are uniformly Lipschitz as
\begin{equation*}%\label{eq:C_L_B}
\tilde \I_{L,B} := \{\psi \in \I_L: \|\psi\|_\mathfrak{X} \le B\}.
\end{equation*} 

Using Theorem 3.2 of \citet{GS13} (also see \citet{Bronshtein76}) we know that   
\begin{eqnarray}\label{difflip.eq}
   \log N \left(\tilde \I_{L, B}, \epsilon; L_{\infty} \right)  \leq c \left(\frac{B + d L}{\epsilon}  \right)^{d/2},
  \end{eqnarray}
for all $0 < \epsilon \leq \epsilon_0 (B + d L)$, where $\epsilon_0 >0$ is a fixed constant and $L_\infty$ is the supremum norm. In the following we denote by $\I_{L}$ and $\I_{L,B}$ the convex sets (in $\R^n$) of all evaluations (at the data points) of functions in $\I_{L}$ and $\I_{L,B}$, respectively; cf.~\eqref{eq:CvxCone}. Thus, we can say that $\hat \bt_{I,L}$ is the projection of $\mathbf{Y}$ on $\I_{L}$. Let $\hat \bt_{I,L,B} $ denote the projection of $\mathbf{Y}$ onto $\I_{L,B}$. We now use Theorem 4.8 of \citet{vdG00} to show that 
\begin{equation}\label{eq:L_2RateBdLip}
	\|\hat \bt_{I,L,B} - \bt_0\|^2 = o_p(n).
\end{equation}
Note that equation (4.26) in \citet{vdG00} is trivially satisfied as we have sub-gaussian errors and equation (4.27) easily follows from \eqref{difflip.eq}. 

Denote the $i$-th coordinate of a vector $\mathbf{b} \in \R^n$ by $b^{(i)}$, for $i=1,\ldots,n$. Define the event $A_n := \{ \max_{i=1,\ldots,n}|\hat \bt_{I,L}^{(i)}|\le B_0 \}$. Next we show that there exists $B_0>0$ such that 
\begin{equation}\label{eq:A_n}
\P(A_n) \rightarrow 1, \qquad \mbox{as } n \rightarrow \infty.
\end{equation} 

As $\hat \bt_{I,L}$ is a projection on the closed convex set $ \I_L$, we have $$\langle \mathbf{Y} - \hat \bt_{I,L}, \B \gamma - \hat \bt_{I,L} \rangle \le 0, \qquad \mbox{for all } \B \gamma \in \I_L.$$ Letting $\mathbf{e} := (1,1,\ldots,1)^\top \in \R^n$, note that for any $c \in \R$, $c  \mathbf{e} \in \I_L$. Hence, $$ \langle \mathbf{Y} - \hat \bt_{I,L}, c  \mathbf{e} - \hat \bt_{I,L} \rangle  = c \langle \mathbf{Y} - \hat \bt_{I,L}, \mathbf{e} \rangle - \langle \mathbf{Y} - \hat \bt_{I,L}, \hat \bt_{I,L} \rangle \le 0, \qquad \mbox{ for all } c \in \R,$$ and thus $\langle \mathbf{Y} - \hat \bt_{I,L}, \mathbf{e} \rangle = 0,$ i.e., $n \bar Y = \sum_{i=1}^n Y_i = \sum_{i=1}^n \hat \bt_{I,L}^{(i)}$. Now, for any $i \in \{1,2,\ldots,n\}$,
\begin{eqnarray*}
	|\hat \bt_{I,L}^{(i)}| & \le & |\hat \bt_{I,L}^{(i)} - \bar Y| + |\bar Y|  \;\; =\;\; \left| \hat \bt_{I,L}^{(i)} - \frac{1}{n} \sum_{j=1}^n\hat \bt_{I,L}^{(j)} \right| + |\bar Y| \\
	& \le & \frac{1}{n} \sum_{j=1}^n  \left|\hat \bt_{I,L}^{(i)} - \hat \bt_{I,L}^{(j)} \right| + |\bar Y| \;\; \le \;\; \frac{L}{n} \sum_{j=1}^n  \|\B x_i - \B x_j \| + \|\phi_0\|_\mathfrak{X} + |\bar \epsilon| \\
	& \le & \sqrt{d} L + \kappa +1 =: B_0,
\end{eqnarray*}
a.s.~for large enough $n$, where we have used the fact that $\|\B x_i - \B x_j \| \le \sqrt{d}$, $\|\phi\|_\mathfrak{X} < \kappa$ for some $\kappa >0$, and that $\bar \epsilon = \sum_{i=1}^n \epsilon_i/n \rightarrow 0$ a.s.

As $\I_{L,B_0} \subset \I_L$, we trivially have $$\|  \mathbf{Y} - \hat \bt_{I,L} \|^2 \le \| \mathbf{Y} - \hat \bt_{I,L,B_0}\|^2.$$ If $A_n$ happens, $\hat \bt_{I,L} \in \I_{L,B_0}$, and thus, $$\|\mathbf{Y} - \hat \bt_{I,L,B_0} \|^2 = \argmin_{\theta \in \tilde \I_{L,B_0}} \|\mathbf{Y} - \bt \|^2 \le \| \mathbf{Y} - \hat \bt_{I,L}\|^2.$$
From the last two inequalities it follows that if $A_n$ occurs, then $\hat \bt_{I,L} = \hat \bt_{I,L,B_0}$, as $\hat \bt_{I,L,B_0}$ is the unique minimizer. Now using \eqref{eq:A_n}, \eqref{eq:L_2RateLip} immediately follows from \eqref{eq:L_2RateBdLip}. \qed

\bibliographystyle{apalike}
\bibliography{shapes}

\end{document}